\documentclass[11pt]{article}
\usepackage[utf8]{inputenc}
\usepackage[T1]{fontenc}
\usepackage{amsthm, amsmath}
\usepackage{amsfonts}
\usepackage[margin=1in]{geometry}
\usepackage[algo2e,ruled,noend,resetcount,linesnumbered]{algorithm2e}
\usepackage{graphicx}
\usepackage{comment}
\usepackage{nicematrix}
\usepackage{tikz}
\usepackage[shortlabels]{enumitem}
\usepackage{thmtools}
\usepackage{thm-restate}
\usepackage{algorithm, algorithmicx, algpseudocode} 
\usepackage{array}
\usepackage{stmaryrd}
\usepackage[dvipsnames]{xcolor}
\usepackage[colorlinks,citecolor=blue,linkcolor=blue,urlcolor=red,pagebackref]{hyperref}
\usepackage[capitalise]{cleveref}
\usepackage{xspace}
\usepackage{caption}
\usepackage{subcaption}

\usetikzlibrary{decorations.pathreplacing}
\usetikzlibrary{arrows}
\usetikzlibrary{math}

\theoremstyle{plain}
\newtheorem{theorem}{Theorem}[section]

\newtheorem{proposition}[theorem]{Proposition}
\newtheorem{lemma}[theorem]{Lemma}
\newtheorem{claim}[theorem]{Claim}

\newtheorem{hypothesis}[theorem]{Hypothesis}
\newtheorem{fact}[theorem]{Fact}

\theoremstyle{definition}
\newtheorem{definition}[theorem]{Definition}

\theoremstyle{remark}

\newcommand{\cC}{\mathcal{C}}

\newcommand{\R}{\mathbb{R}}

\newcommand{\kOV}{k\textup{-}\mathsf{OV}}
\newcommand{\kOVf}[1]{#1\textup{-}\mathsf{OV}}
\newcommand{\kSAT}{k\textup{-}\mathsf{SAT}}
\newcommand{\OV}{\textsf{OV}}

\newcommand{\SETH}{\textup{\textsf{SETH}}}

\newcommand{\eps}{\varepsilon}

\newcommand{\set}[1]{\left\{#1\right\}}
\newcommand{\otherwise}{\mathrm{ o/w }}
\newcommand{\given}{\mathrm{~ s.t. ~}}

\newcommand{\tO}[1]{\tilde{O}(#1)}
\newcommand{\poly}{\mathrm{ poly}}

\newcommand{\bigO}[1]{O\left(#1\right)}

\newcommand{\bigOm}[1]{\Omega\left(#1\right)}

\newcommand{\ptime}{\textbf{P}}
\newcommand{\np}{\textbf{NP}}

\newcommand{\ind}[1]{\mathbf{1}\left[#1\right]}

\newcommand{\mm}{\mathsf{MatMul}}

\newcommand{\transpose}{\top}
\newcommand{\T}{\transpose}

\newcommand{\diag}{\mathrm{diag}}

\newcommand{\norm}[1]{\left\lVert#1\right\rVert}

\newcommand{\pderiv}[2]{\frac{\partial #1}{\partial #2}}

\newcommand{\Att}{\mathrm{Att}}

\newcommand{\attn}[2]{\Att^{#2}_{#1}}

\newcommand{\din}{{d_{\textup{in}}}}
\newcommand{\dout}{{d_{\textup{out}}}}
\newcommand{\dhid}{{d_{\textup{hid}}}}

\newcommand{\evec}[1]{e_{#1}}

\newcommand{\sm}{\mathrm{softmax}}
\newcommand{\hm}{\mathrm{hardmax}}
\newcommand{\imax}{I_{\max}}
\newcommand{\hard}{\mathrm{hard}}

\newcommand{\eAC}{{\rm eAC}\xspace}
\newcommand{\AC}{{\rm AC}\xspace}

\def\colorful{1}

\ifnum\colorful=1
\newcommand{\todo}[1]{\textcolor{red}{[TODO: #1]}}
\newcommand{\barna}[1]{\textcolor{red}{[Barna: #1]}}
\newcommand{\chris}[1]{\textcolor{purple}{[Chris: #1]}}
\newcommand{\yinzhan}[1]{\textcolor{orange}{[Yinzhan: #1]}}
\newcommand{\hantao}[1]{\textcolor{cyan}{[Hantao: #1]}}

\else
\newcommand{\todo}[1]{}
\newcommand{\barna}[1]{}
\newcommand{\chris}[1]{}
\newcommand{\yinzhan}[1]{}
\newcommand{\hantao}[1]{}

\fi

\title{On the Computational Hardness of Transformers}

\author{
\begin{tabular}[t]{@{}c@{\hspace{1cm}}c@{}}
\parbox{6.5cm}{\centering
Barna Saha\thanks{Supported by NSF HDR TRIPODS Phase II grant 2217058 (EnCORE Institute).}  \\
University of California San Diego  \\
\texttt{bsaha@ucsd.edu}
}
&
\parbox{6.5cm}{\centering
Yinzhan Xu\footnotemark[1]  \\
University of California San Diego  \\
\texttt{xyzhan@ucsd.edu}
}
\\[1cm]
\parbox{6.5cm}{\centering
Christopher Ye\footnotemark[1]  \\
University of California San Diego  \\
\texttt{czye@ucsd.edu}
}
&
\parbox{6.5cm}{\centering
Hantao Yu\thanks{Supported by NSF grant CCF2008733, ONR grant N00014-22-1-2713, NSF Grant CCF-2238221 and a Columbia Center of Artificial Intelligence Technology grant.}  \\
Columbia University  \\
\texttt{hantao.yu@columbia.edu}
}
\end{tabular}
}
\date{}

\begin{document}

\maketitle

\begin{abstract}
    The transformer architecture has revolutionized modern AI across language, vision, and beyond. It consists of $L$ layers of multi-head attention, where each layer runs $H$ attention heads in parallel and feeds the combined output to the subsequent layer. In attention, each token within an input of length $N$ is represented by an embedding vector of dimension $m$. Computationally, an attention mechanism primarily involves multiplying three $N \times m$ matrices, while applying a softmax operation to the intermediate product of the first two matrices. A significant body of work has been devoted to analyzing the time complexity of attention, leading to several recent advances.

On the other hand, known algorithms for transformers compute each attention head independently. This raises a fundamental question that has recurred throughout theoretical computer science under the guise of ``direct sum'' problems: can multiple instances of the same problem be solved more efficiently than solving each instance separately? Many answers to this question, both positive and negative, have arisen in fields spanning communication complexity and algorithm design. Thus, a key challenge in understanding the computational hardness of transformers is to determine whether their computation can be performed more efficiently than $LH$ independent evaluations of attention.

In this paper, we resolve this question in the negative, and give the first non-trivial computational lower bounds for multi-head multi-layer transformers. In the small embedding regime ($m = N^{o(1)}$), computing $LH$ attention heads separately takes $LHN^{2 + o(1)}$ time. We establish that this is essentially optimal under the Strong Exponential Time Hypothesis (SETH). In the large embedding regime ($m = N$), one can compute $LH$ attention heads separately using $LHN^{\omega + o(1)}$ arithmetic operations (plus exponents), where $\omega$ is the matrix multiplication exponent. We establish that this is optimal, by showing that $LHN^{\omega - o(1)}$ arithmetic operations are necessary when $\omega > 2$. Our lower bound in the large embedding regime relies on a novel application of the Baur-Strassen theorem, a powerful algorithmic tool underpinning the famous backpropagation algorithm.

\end{abstract}

\newpage

\section{Introduction}
\label{sec: intro}

The Transformer architecture \cite{DBLP:conf/nips/VaswaniSPUJGKP17} is used for a wide range of tasks such as natural language processing~\cite{radford2019language} and computer vision \cite{dosovitskiy2020image, kirillov2023segment}, and has achieved state-of-the-art performance for numerous applications. The crucial building block of transformers is the attention head, which captures  pairwise relationships between input tokens and thus is desirable for processing language inputs. Concretely, given input $X \in \mathbb{R}^{N \times m}$, the attention head consists of query, key and value embedding maps $Q,K,V: \mathbb{R}^{N \times m} \rightarrow \mathbb{R}^{N \times m}$, where $m$ is known as the embedding dimension, and computes\[
\Att_{Q,K,V}(X) = \sm(Q(X)K(X)^{\top})V(X) \in \mathbb{R}^{N \times m},
\] where the softmax operator is defined as
\[
\sm(v) = \frac{1}{\sum_{i=1}^{N}\exp(v_i)}\cdot(\exp(v_1),\ldots,\exp(v_N))
\] and is applied row-wise on matrices. A transformer consists of $L$ attention layers, where each attention layer consists of $H$ parallel attention heads that takes the output $X$ of the previous layer as input, sends it to all attention heads, and aggregate\footnote{Attention outputs can either be aggregated by sum or by concatenation. In \Cref{app:head-aggregation}, we reduce summation aggregation to concatenation aggregation. Thus, we may consider summation aggregation without loss of generality.} the outputs of all attention heads and $X$ before finally applying a row-wise multi-layer perceptron (MLP) function.

While the transformer is undoubtedly powerful, it is also computationally expensive, with attention complexity scaling quadratically with respect to input length $N$. 
The trivial algorithm for attention requires $O(N^2m)$ time,\footnote{Following practical values, we always assume $m \leq O(N)$.} and the trivial algorithm for transformer requires $O(LHN^2m)$ time, even before accounting for the time required to compute embedding and MLP maps.\footnote{The time needed to compute MLPs eventually depends on its hidden dimension $\dhid$. For inputs in $\R^m$, the time for computing MLP in practice should roughly be $O(m\dhid)$. 
In this paper, we obtain lower bounds for transformers that do not have MLPs and whose embedding maps can be computed in linear time $O(Nm)$. 
In \Cref{app:mlp-examples}, we describe how to generalize our lower bounds to transformers with popular MLPs.
Note that this only strengthens the lower bound, as these are specific instances of transformers.}
It has been proved that $N^{2 - o(1)}$ time is necessary for computing attention when $m = \Omega(\log N)$ assuming the Strong Exponential Time Hypothesis (\SETH) \cite{keles2023computational, alman2023fastattentionboundedentries}, while the trivial algorithm of computing matrix products explicitly is optimal for all $m$ under a generalization of the $\OV$ Hypothesis \cite{gupta2026subquadratic}.
Due to the lack of fast attention algorithms, numerous works have tried to address this bottleneck by proposing subquadratic alternatives to attention mechanism \cite{DBLP:conf/iclr/KitaevKL20, DBLP:conf/iclr/ChoromanskiLDSG21, DBLP:conf/nips/ZaheerGDAAOPRWY20, liu2025fastattentionmechanismstale}. 
However, all these subquadratic alternatives suffer from accuracy drops as the model scales up, and it has been proved that no subquadratic alternative can capture all pairwise relationship up to some approximation~\cite{DBLP:conf/iclr/AlmanY25}. 
Consequently, our work is primarily concerned with the following question. 

\begin{center}
    {\bf Question 1:} {\it How efficiently can a standard transformer model be computed?}
\end{center}  

While \cite{keles2023computational, alman2023fastattentionboundedentries, gupta2026subquadratic} address this problem for a single attention head, we are instead concerned with transformers, which consist of many layers of multi-head attention.
Proving lower bounds against transformers  
has been a critical challenge.
Indeed, transformers appear to be significantly more complex  than attention heads. Even while considering expressivity (i.e., which functions can be computed), there are relatively plentiful techniques to prove lower bounds against attention or single layer transformers (see e.g. one-way communication \cite{DBLP:conf/nips/SanfordHT23} and split-VC dimension~\cite{kozachinskiy2025strassenattentionsplitvc}), but only a recent breakthrough established strong lower bounds against decoder-only transformers with many layers \cite{chen2024theoretical}.

Understanding the complexity of transformers is further motivated by known lower bounds on specific tasks that we hope transformers can handle.
There are many simple, fundamental problems that require large  transformers.
For example, it is conjectured that any transformer that is capable of computing 3-SUM requires polynomial size (i.e. $mHL = N^{\bigOm{1}}$)~\cite{DBLP:conf/nips/SanfordHT23, DBLP:conf/icml/Sanford0T24}.
Furthermore, shallow decoder transformers cannot compute function composition without polynomial width \cite{chen2024theoretical}.
If we nevertheless want to use transformers to compute such functions (e.g., 3-SUM naturally captures interactions between triples of input tokens), it is crucial that large transformers can be computed as efficiently as possible.

Beyond its obvious practical motivation, studying the complexity of transformers is also an interesting theoretical question on its own. Understanding the complexity of transformers is essentially understanding whether we can compute many attention heads faster than computing them one by one. 
One can ask such questions more generally:
\begin{center}
    {\bf Question 2:} {\it Can the cost of solving many copies of the same problem simultaneously be amortized more efficiently than solving each copy separately?}
\end{center}  
This is a recurring question in theoretical computer science, typically known as the ``direct sum'' problem, and there are numerous positive and negative examples to this question. 

On the positive side, a famous result in communication complexity gives a partial function with complexity $\Theta(\log n)$ but constant $O(1)$ amortized complexity when computing many instances of the function at the same time \cite{feder1995amortized}.
Recently, a breakthrough obtained a similar result for total functions \cite{mackenzie2025refuting}.
A degree $n$ polynomial can be evaluated on $n$ points in $\tO{n}$ time \cite{fiduccia1972polynomial} in contrast to the naive $\tO{n^2}$ algorithm. 
Similarly, the product of an $n \times n$ matrix with $n$ distinct vectors can be computed in $O(n^{2.372})$ time \cite{DBLP:conf/soda/AlmanDWXXZ25} in contrast to $O(n^3)$ which is required to compute them all separately.
Moreover, to compute a transformer we do not need to compute each attention head explicitly. 
For example, the output of a layer is the \emph{sum} of $H$ attention heads.
For matrix products, given $n$ pairs of $n \times n$ matrices $\set{A_{i}, B_{i}}$, the sum of products, $\sum_{i} A_{i} B_{i}$ can be computed in $O(n^{3.251})$ time, much faster than computing each product separately, which takes $O(n^{3.372})$ time \cite{DBLP:conf/soda/AlmanDWXXZ25}.

There is an equal abundance of results on the negative side.
For example, monotone functions with independent inputs cannot be computed more efficiently than computing each separately~\cite{galbiati1981complexity}.
Similarly, \cite{bshouty1989extended} give a direct sum theorem for certain classes of quadratic systems. 
\cite{blais2019optimal} also obtain a strong direct sum theorem for randomized query complexity.
\cite{molinaro2013beating} use strong direct sum theorems to obtain distributed and streaming lower bounds for various sketching and estimation tasks.
In contrast to the matrix-vector multiplication problem, it is widely hypothesized that the $(\min, +)$-product of an $n \times n$ matrix with $n$ vectors cannot be computed in strongly subcubic time \cite{williams2018subcubic}.
Question 2 falls naturally within this long line of research: asking whether transformers admit a direct sum theorem.

In this work, we essentially resolve Question 1 (or Question 2 in the case of attention and transformer) in the negative: we give the first non-trivial lower bounds for multi-head multi-layer transformer  computation and show that the naive algorithm of computing each attention head separately is essentially optimal.

\subsection{Our Contributions}

We give stronger lower bounds for both the small embedding dimension ($m = N^{o(1)}$) case and the large embedding dimension ($m = N$) case.
In the small embedding dimension case, the naive algorithm runs in $LHN^{2 + o(1)}$ time while the previously known conditional lower bounds on attention~\cite{alman2023fastattentionboundedentries, gupta2026subquadratic} rule out algorithms in $LHN^{1 - o(1)} + N^{2 - o(1)}$ time.
We show a significantly improved conditional lower bound, matching the naive algorithm. 

\begin{theorem}[Informal \Cref{thm:small-embedding-formal}]
    \label{thm:small-embedding}
    Under the $\kOVf{3}$ Hypothesis or \SETH, any algorithm computing $L$ layer, $H$ head transformers with summation aggregation and embedding dimension $m = \bigOm{\log N}$ requires
    \begin{equation*}
        LHN^{2 - o(1)} 
    \end{equation*}
    time.
    The lower bound holds even if the transformer has no MLPs.
\end{theorem}

The $\kOVf{3}$ Hypothesis states that finding a triplet of orthogonal vectors among a set of $n$ vectors from $\{0, 1\}^{\Theta(\log n)}$ requires $n^{3 - o(1)}$ time.
The $\kOVf{3}$ Hypothesis is a special case of the $\kOV$ Hypothesis, which has been used in many lower bounds e.g.~\cite{backurs2018towards, li2021settling, bonnet20224, alman_et_al:LIPIcs.ITCS.2024.4, dalirrooyfard2025hardness}.
Furthermore, it is a consequence of the Strong Exponential Time Hypothesis (\SETH)~\cite{williams2004ovc}, a strengthening of $\ptime \neq \np$ and a popular conjecture in fine-grained complexity that states that there is no $2^{(1 - \eps)n}$ algorithm for satisfiability on $n$ variables \cite{impagliazzo2001seth} for any constant $\eps > 0$.
\SETH~implies essentially optimal lower bounds for problems including edit distance \cite{backurs2015edit}, graph diameter \cite{roditty2013fast}, dynamic programming problems \cite{alman_et_al:LIPIcs.ITCS.2024.4}, approximate closest pair problems \cite{abboud2017distributed}, subset sum~\cite{DBLP:conf/soda/AbboudBHS19} and many other problems.

Even though many previous theoretical works model transformer embedding dimension to be roughly $O(\log N)$, recent work also considers the scenario where embedding dimension is close to $O(n)$. 
\cite{yehudai2025depthwidthtradeoffsalgorithmicreasoning} finds that for many practical graph datasets, the number of graphs (input length) could be well smaller than the embedding dimension of the transformers that are used in practice. For larger embedding dimension $m$ the bound in \Cref{thm:small-embedding} does not get better. 
For example, when the embedding dimension $m$ is the same as the input length (i.e. $m = N$), the best running time becomes worse: the naive upper bound is $O(LHN^{2}m) = O(LHN^3)$, and with fast matrix multiplication, the upper bound can be improved to $LHN^{\omega + o(1)}$, where $\omega$ denotes the matrix multiplication exponent.
An alternative lower bound can be obtained by using the existing conditional lower bound of $N^{\omega - o(1)}$ for a single attention computation \cite{gupta2026subquadratic}. 
Nevertheless, there remains a large gap between the best lower bound $\max\{LH N^{2-o(1)}, N^{\omega - o(1)}\}$ and the best upper bound $LH N^{\omega + o(1)}$, and the time complexity of transformer in this setting remains poorly understood.

In this work, we show that computing transformers is as hard as  computing $\Theta(LH)$ instances of matrix multiplication, thus obtaining a matching $LHN^{\omega - o(1)}$ lower bound without any additional hypotheses, as long as $\omega > 2$.

Our result is shown in an extended version of the arithmetic circuit model, which is a generalization of the well-studied arithmetic circuit model (see e.g., \cite{burgisser2013algebraic} for a book on arithmetic circuits). In the standard arithmetic circuit model, an algorithm only uses standard arithmetic operations $\set{+, -, \times, /}$.

Since attention computation involves taking exponents, it is not clear how to compute attention at all using a standard arithmetic circuit, i.e. a standard arithmetic circuit cannot perform attention computation as simulating exponential function could already be hard enough. 
Thus, it is natural to also allow exponential gates and logarithmic gates in the arithmetic circuits we consider, and we call them \emph{extended arithmetic circuits} (eACs). We remark that lower bounds against extended arithmetic circuits are stronger than lower bounds against arithmetic circuits, as any arithmetic circuit is an extended arithmetic circuit.
Furthermore, all known algorithms for transformer computation exist in the extended arithmetic circuit model.

With fast matrix multiplication, it is easy to obtain an extended arithmetic circuit of size $LHN^{\omega+o(1)}$ that simulates $H$-head $L$-layer transformer inference with embedding dimension $N$.
We show that this is essentially optimal.
In particular, any extended arithmetic circuit computing transformers requires size $LHN^{\omega - o(1)}$.

\begin{restatable}[Informal \Cref{thm:linear-embedding-formal}]{theorem}{LinearEmbeddingFormal}
    \label{thm:linear-embedding}
    Any extended arithmetic circuit computing $L$-layer $H$-head transformers with summation aggregation and embedding dimension $m = \Omega(N)$ has size at least
    \begin{equation*}
        L H N^{\omega - o(1)}
    \end{equation*} 
    when $\omega > 2$. 
    The lower bound holds even if the transformer has no MLPs.
\end{restatable} 

We remark that when $\omega = 2$ (i.e. transformers with $m = O(N)$ can be computed in $LHN^{2 + o(1)}$ time), we can instead apply \Cref{thm:small-embedding} to obtain a matching $LHN^{2 - o(1)}$ lower bound. For extension of our results to concatenation aggregation, see \Cref{app:head-aggregation}. 

Our lower bound on transformers shows that when embedding dimension is large, no significant improvement is possible outside fast matrix multiplication for computing transformers in the extended arithmetic circuit model. 
We provide a new way of proving lower bounds for multilayer transformer, which has been a challenging direction in this literature.

\subsection{Technical Overview}

In this section, we give an overview of our techniques.

\paragraph{Small Embedding Dimension.}

We begin with our lower bound in the small embedding dimension regime $m = N^{o(1)}$.
Consider a $\kOVf{3}$ instance $A = \{a_1, \ldots, a_N\}, B = \{b_1, \ldots, b_N\}, C = \{c_1, \ldots, c_{LH}\} \subset \set{0, 1}^{d}$ for some $d = \Theta(\log n)$. 
Note that the three sets have different sizes, and we call it an unbalanced $\kOVf{3}$ instance. 
It is known that under the standard $\kOVf{3}$ Hypothesis, such instances require $LHN^{2-o(1)}$ time. 
We construct a transformer $T_{C}$ with $L$ layers, $H$ heads per layer, and embedding dimension $O(d)$ such that $T_{C}(A, B)$ (with some simple postprocessing) determines whether $(A, B, C)$ contains an orthogonal triple.
First, we specify the input to the transformer:
\begin{equation*}
    X = \begin{pmatrix}
        a_{1}^{\T} & b_{1}^{\T} & 1 \\
        \vdots & \vdots & \vdots \\
        a_{N}^{\T} & b_{N}^{\T} & 1 \\
        0_{1 \times d} & 0_{1 \times d} & 2
    \end{pmatrix} \text{.}
\end{equation*}
Except the last token, each token consists of one vector from $A$, one vector from $B$, and the constant $1$. 
We now describe $T_{C}$ which has input dimension $2d + 1$.
$T_{C}$ will consist of $L$ layers and $H$ heads per layer and embedding dimension $m = 2d + 2$.
For every head, we set $V$ to be $0$ in the first $2d + 1$ columns, so that the output of the $\ell$-th layer, denoted as $X^{(\ell)}$, contains $X = X^{(0)}$ in the first $2d + 1$ columns.
In particular, we may assume that the input to every attention head contains $X$ in the first $2d + 1$ columns.

We now describe the embedding map of each head.
Since $|C| = LH$, we can index $C = \set{c_{h, \ell}}$ and identify each vector in $C$ with one head in the transformer.
Consider the $h$-th head in the $\ell$-th layer.
We define embedding maps $q_{h, \ell}, k_{h, \ell}, v_{h, \ell}: \R^{m} \rightarrow \R^{m}$ as follows.
Note that since our maps do not depend on the $m$-th coordinate, the following maps remain well defined for the first layer which is a map $\R^{m - 1} \rightarrow \R^{m}$.
\begin{align*}
    q_{h, \ell}&: u \mapsto - (u_{1}, \dots, u_{d}, 0, \dots, 0), \\
    k_{h, \ell}&: u \mapsto (u_{d + 1} \cdot c_{h, \ell}[1], \dots, u_{2d} \cdot c_{h, \ell}[d], 0, \dots, 0), \\
    v_{h, \ell}&: u \mapsto (0, \dots, 0, u_{2d + 1}) \text{.}
\end{align*}
Note that we can implement the above maps with appropriate weight matrices $W_{Q, h, \ell}, W_{K, h, \ell}, W_{V, h, \ell} \in \R^{m \times m}$ for all $h, \ell$.
Applying the above map to all $N$ tokens, we obtain
\begin{equation*}
    Q_{h, \ell}(X) = \begin{pmatrix}
        - a_{1}^{\T} & 0_{1 \times (d + 2)} \\
        \vdots & \vdots \\
        - a_{N}^{\T} & 0_{1 \times (d + 2)} \\
        0_{1 \times d} & 0_{1 \times (d + 2)}
    \end{pmatrix}, \quad
    K_{h, \ell}(X) = \begin{pmatrix}
        b_{1}^{\T} \odot c_{h, \ell} & 0_{1 \times (d + 2)} \\
        \vdots & \vdots \\
        b_{N}^{\T} \odot c_{h, \ell} & 0_{1 \times (d + 2)} \\
        0_{1 \times d} & 0_{1 \times (d + 2)}
    \end{pmatrix}, \quad 
    V_{h, \ell}(X) = \begin{pmatrix}
        0_{1 \times (2d + 1)} & 1 \\
        \vdots & \vdots \\
        0_{1 \times (2d + 1)} & 1 \\
        0_{1 \times (2d + 1)} & 2
    \end{pmatrix}
\end{equation*}
where $\odot$ denotes the entry-wise product of two vectors and $0_{a \times b}$ (resp. $1_{a \times b}$) denotes an $a \times b$ submatrix filled with constant $0$ (resp. $1$).
Thus, the output of the $h$-th head in the $\ell$-th layer is $\sm(Q_{h, \ell}(X) K_{h, \ell}(X)^{\T})V_{h, \ell}(X)$.
To simplify the exposition slightly, we instead consider hardmax attention.
The $\hm$ function maps a vector $v$ to a uniform distribution over its maximum elements.
That is, for a vector $v$ with $\imax(v) = \set{i \given v_{i} = \max_j v_{j}}$, $\hm(v)_{i} = \frac{1}{|\imax(v)|}$ if $i \in \imax(v)$ and $\hm(v)_{i} = 0$ otherwise.
A hardmax attention head then outputs $Z_{h, \ell} := \hm(Q_{h, \ell}(X) K_{h, \ell}(X)^{\T})V_{h, \ell}(X)$ where $\hm$ is applied row-wise.
It is known that hardmax attention can be approximated with softmax attention \cite{DBLP:conf/icml/Sanford0T24}, so it is essentially without loss of generality to consider hardmax attention. 
In fact, by scaling the entries of $Q_{h, \ell}$ and $K_{h, \ell}$ by an appropriate factor on the order of $O(\log NHL)$, we can ensure that each attention head outputs $Y_{h, \ell}$ such that
\begin{equation}
    \label{eq:softmax-error-overview}
    \norm{Y_{h, \ell} - \hm(Q_{h, \ell}(X) K_{h, \ell}(X)^{\T})V_{h, \ell}(X)}_{\infty} \ll \frac{1}{\poly(NHL)} \text{.}
\end{equation}

Then, denoting $Z_{h, \ell}$ as the last column of $\hm(Q_{h, \ell}(X) K_{h, \ell}(X)^{\T})V_{h, \ell}(X)$, we claim $Z_{h, \ell}$ is a vector whose $i$-th entry is $2$ if $a_{i}$ and $c_{h, \ell}$ are not in any orthogonal triple (i.e. there are no $b_{j}$ such that $a_{i}, b_{j}, c_{h, \ell}$ are orthogonal) and at most $\frac{3}{2}$ if $a_{i}, c_{h, \ell}$ participate in an orthogonal triple.
This follows as in the former case, the set of maximum indices in the $i$-th row is only $N + 1$ where $V_{h, \ell}(X)$ contains a $2$; in the latter case, the set of maximum indices contains at least one other entry, and thus $\hm$ averages over at least one row where $V$ contains $1$.
Summing over all $h, \ell, i$, we have $\sum_{h, \ell} \sum_{i} Z_{h, \ell}[i] = 2 NHL$ if there is no orthogonal triple and at most $2 NHL - \frac{1}{2}$ if there is.
From \eqref{eq:softmax-error-overview}, we can use $\sum_{h, \ell} \sum_{i} Y_{h, \ell}[i, m]$ to determine if $A, B, C$ contains an orthogonal triple.

\paragraph{Large Embedding Dimension.}

We now describe our lower bound for extended arithmetic circuits that compute transformers with large embedding dimension ($m = N$).
Recall that computing a single attention head is in $N^{\omega+o(1)}$ time, as we require two matrix products, one to compute $QK^{\T}$ and one to compute the product of $\sm(QK^{\T})$ and $V$.
It is then easy to see that we can use $O(LH)$ instances of matrix multiplication to compute a transformer with $L$ layers, $H$ heads, and embedding dimension $N$.

The matching lower bound is much more challenging.
While it is not too difficult to show that we can reduce matrix multiplication to a single attention head with embedding dimension $N$ (one can, for example, set $Q \gets A$, $K \gets B^{\T}$ and $V \gets I$ so that the attention head outputs $\sm(QK^{\T})$. We can use a simple denormalization trick to obtain obtain $\exp(QK^{\T})$ instead (see \Cref{app:omitted-proofs}), from which we can compute $AB^T$, 
it is not clear how to use a single transformer to simultaneously compute many independent instances of matrix multiplication.
Consider as an example a transformer with a single layer and $H$ heads.
The output of this transformer is the \emph{sum} of all attention heads in this layer.
As discussed earlier, the sum of matrix products can be computed more efficiently than computing each matrix product separately (e.g. the sum of $N$ matrix products can be computed in $\bigO{N^{3.251}}$ time while computing each product separately requires $\bigO{N^{3.372}}$ time.
Thus, if we directly repeat the reduction of a single matrix product to a single attention head for all $LH$ heads in our transformer, the output of the transformer only computes the sum of $LH$ attention outputs, which has dimension $\Theta(N^2)$, in contrast to the $\Theta(LHN^2)$ dimension required to express $LH$ distinct matrix products.

The key tool we require is the Baur-Strassen theorem \cite{DBLP:journals/tcs/BaurS83}, which was originally designed to show lower bounds in the arithmetic circuit model. 
Perhaps the most famous application of the Baur-Strassen theorem is the backpropagation algorithm, a fundamental building block of efficient training of neural networks \cite{rumelhart1986learning}. 
The Baur-Strassen theorem has also found applications in algorithm design \cite{cygan2015algorithmic,fischer2024deterministic}. We give a new application of the Baur-Strassen theorem: to obtain lower bounds for transformers.
The Baur-Strassen theorem states that given an arithmetic circuit of size $s$ computing a function $f: \R^{n} \rightarrow \R$, there is an arithmetic circuit of size $O(s)$ computing all the partial derivatives of $f$ with respect to its $n$ inputs.
In particular, Baur-Strassen gives us a powerful tool to extract more information from the partial computations within our algorithm, which provides an avenue to obtain the required $\Omega(LHN^2)$ dimensional output.

In standard settings, an arithmetic circuit (\Cref{def:arithmetic-circuit}) consists of $+, -, \times, /$ gates, and the size of a circuit is the total number of gates.
As discussed previously, since the key non-linearity in attention is the softmax function which requires exponentiation, we equip our arithmetic circuits with $\exp$ and $\ln$ gates, and call them \emph{extended arithmetic circuits}.
Our first step is to give an extension of the Baur-Strassen theorem: if there is an extended arithmetic circuit of size $s$ computing a function $f: \R^{n} \rightarrow \R$, then there is an extended arithmetic circuit of size $O(s)$ computing all the partial derivatives of $f$.

Equipped with this tool, we now describe our lower bound construction.
For simplicity, we describe our construction with what we call denormalized attention heads, which replaces the row-wise $\sm$ with entry-wise $\exp$.
In particular, given $Q(X), K(X), V(X)$, a denormalized attention head outputs $\exp(Q(X)K(X)^{\T})V(X)$ (see \Cref{app:omitted-proofs} for a more detailed discussion on denormalized attention heads).
In the proof of \Cref{thm:linear-embedding} (see \Cref{thm:linear-embedding-formal}), we modify the construction given below to ensure that the matrix product can be extracted from a standard attention head (with normalization).
Equipped with denormalized attention heads, we may construct a transformer $T$ that outputs a vector $T(X) \in \R^{N}$ with entries $T(X)_{i} = \sum_{k = 1}^{LH} \sum_{j = 1}^{N} \exp((A_{k} B_{k}^{\T})_{ij})$ where $A_{k}, B_{k} \in \R^{N \times m}$ are arbitrary $N \times m$ matrices for $k \in [LH]$.
Our transformer simply consists of $H$ denormalized attention heads in each layer, where the $h$-th head in the $\ell$-th layer sets $Q_{h, \ell}(X) \gets A_{(\ell - 1)H + h}$, $K_{h, \ell}(X) \gets B_{(\ell - 1)H + h}$ and $V_{h, \ell}(X) \gets 1_{N \times 1}$ is the vector of all $1$. 
(Technically $V_{h, \ell}(X)$ should be an $N \times m$ matrix, but here we do not need rest of the columns and can simply pad them with $0$'s.)
Then, $\exp(Q_{h, \ell}(X) K_{h, \ell}(X)^{\T}) V_{h, \ell}(X) = \exp(A_{(\ell - 1)H + h} B_{(\ell - 1)H + h}^{\T}) \cdot 1_{N \times 1}$ is the vector consisting of row sums of $\exp(A_{(\ell - 1)H + h} B_{(\ell - 1)H + h}^{\T})$.
We then conclude by summing up the outputs of all heads in each layer, and proceed inductively across all $L$ layers.

Suppose there is an extended arithmetic circuit (\eAC) of size $s$ that computes the transformer $T$. 
We note that while our transformer $T$ has fixed embedding and MLP maps, and thus can be computed by a fixed circuit, the embedding and MLP maps may be parameterized by up to $O(m)$ weights which are additionally provided as input to the circuit. 
See \cref{ssec:transformers} for more detailed discussions on our models.
Given the above transformer, our goal is now to apply the extended Baur-Strassen theorem to extract useful derivatives.
We begin by modifying the input slightly.
Introducing auxiliary input variables $C_{kij}$ for $k \in [LH]$ and $i, j \in [N]$, we define $C_{k}$ to be the $N \times N$ matrix with entries $C_{k}[i, j] = C_{kij}$.
Then, we define for all $k \in [LH]$ matrices $\tilde{A}_{k}, \tilde{B}_{k} \in \R^{N \times 2N}$ where
\begin{equation*}
    \tilde{A}_{k} = \begin{pmatrix}
        A_{k} & C_{k}
    \end{pmatrix},\quad 
    \tilde{B}_{k} = \begin{pmatrix}
        B & I
    \end{pmatrix} \text{.}
\end{equation*}
First, we observe that $\tilde{A}_{k}, \tilde{B}_{k}$ can be computed efficiently in $\bigO{N^2}$ time.

We now describe how to use the transformer $T$ to compute the matrix product $A_{k}B_{k}^{\T}$ for all $k \in [LH]$.
Running the transformer $T$ and summing all entries of the output vector we obtain
\begin{equation*}
    F \gets \sum_{k = 1}^{LH} \sum_{i, j = 1}^{N} \exp\left(\left(\tilde{A}_{k} \tilde{B}_{k}^{\T}\right)_{ij} \right) = \sum_{k = 1}^{LH} \sum_{i, j = 1}^{N} \exp\left(\left(A_{k} B_{k}^{\T} \right)_{ij} + C_{kij} \right) \text{.}
\end{equation*}
Now, consider the derivative of $F$ with respect to $C_{k i j}$.
\begin{align*}
    \pderiv{F}{C_{k i j}} &= \pderiv{}{C_{k i j}} \exp\left(\left(A_{k} B_{k}^{\T} \right)_{i j} + C_{k i j} \right) \\
    &= \exp\left(\left(A_{k} B_{k}^{\T} \right)_{i j} + C_{k i j} \right) \text{.}
\end{align*}
Now, if we set $C_{kij} \gets 0$ for all $k \in [LH]$ and $i, j \in [N]$,
\begin{align*}
    \pderiv{F}{C_{k i j}} &= \exp\left(\left(A_{k} B_{k}^{\T} \right)_{i j} \right)\text{.}
\end{align*}
Thus, by applying the extended Baur-Strassen theorem, we obtain the entries of $\exp\left(A_{k} B_{k}^{\T}\right)$.
From, here, applying $\bigO{LHN^2}$ $\ln$ gates, we obtain the desired $LH$ matrix products.
In particular, we obtain an \eAC of size $\bigO{s + LHN^2}$ that computes $LH$ independent matrix products.

The final ingredient we require is a lower bound on the size of any \eAC that computes $LH$ independent matrix products.
While it is well-known  that a standard arithmetic circuit requires $LHN^{\omega- o(1)}$ size to compute $LH$ independent matrix products (see \Cref{lem:independent-product-lb}), the lower bound for extended arithmetic circuits is a bit more subtle.
At first glance, an extended arithmetic circuit may be significantly more powerful than a standard arithmetic circuit, and may be able to compute matrix products using fewer than $N^{\omega}$ gates.
We show that when restricted to low degree functions, any extended arithmetic circuit of size $s$ can be simulated by a standard arithmetic circuit (without $\exp$ and $\ln$ gates) of size $O(s)$.
That is, extended arithmetic circuits exhibit no advantage when computing low degree functions.
Our result extends that of Strassen \cite{strassen1973vermeidung} (see \cite{DBLP:journals/toc/Blaser13} for an exposition), which states that division gates are not useful for computing low degree functions, and may be of independent interest.

\subsection{Further Related Work}

We review two lines of work in the theoretical study of transformers: (1) the expressivity of transformers, i.e. what functions can transformers compute and learn; (2) complexity of transformers i.e. how efficiently can transformers (or approximations of transformers) be computed?

\paragraph{The Expressivity of Transformers.}

There is a long line of work studying the expressivity of transformers.
Several works, initiated by \cite{hewitt2020rnns, yao2021self}, establish the advantage of transformers over recurrent architectures in tasks such as parsing hierarchical structure \cite{yao2021self}, sparse averaging \cite{DBLP:conf/nips/SanfordHT23}, in context learning \cite{wen2024rnns}, copying \cite{jelassi2024repeat}, and document similarity \cite{DBLP:conf/iclr/AlmanY25}.

Despite their expressive power, lower bounds against transformers show that many functions cannot be computed with small transformers.
Lower bounds against transformers typically rely on communication complexity \cite{DBLP:conf/nips/SanfordHT23, peng2024limitations, DBLP:conf/icml/Sanford0T24} for depth one lower bounds, or reductions to other computational models \cite{merrill2023parallelism, peng2024limitations, DBLP:conf/icml/Sanford0T24}. 
Notably, \cite{DBLP:conf/icml/Sanford0T24} connects the transformer model with the Massively Parallel Computation model by proving an equivalence between transformer and MPC protocol under certain parameters, and thus provide lower bounds on transformers using problems that are hard in MPC. 
In addition, a hard function for transformers with depth $\tO{1}$ was exhibited \cite{chen2024theoretical}.
\cite{hahn2020theoretical} establish that hardmax transformers cannot compute parity or dyck languages.

\paragraph{Algorithms and Hardness for Transformer Computation.}

The exact computation of attention is fairly well understood, with quadratic lower bounds established in the small embedding regime \cite{keles2023computational} and extended to more generic settings \cite{alman2023fastattentionboundedentries, DBLP:conf/iclr/AlmanY25, gupta2026subquadratic}.
In light of these lower bounds, subquadratic algorithms have been developed in the special settings of small entries \cite{alman2023fastattentionboundedentries} and constant embedding dimension \cite{gupta2026subquadratic}.

Given the hardness of exact attention computation, a long line of work has studied efficient approximations of attention, including \cite{DBLP:conf/nips/BrownMRSKDNSSAA20}, Reformer \cite{DBLP:conf/iclr/KitaevKL20}, \cite{DBLP:conf/icml/KatharopoulosV020}, Big Bird \cite{zaheer2020bigbird}, Scatterbrain \cite{DBLP:conf/nips/ChenDWSRR21}, Performer \cite{DBLP:conf/iclr/ChoromanskiLDSG21}, KDEFormer \cite{zandieh2023kdeformer}, HyperAttention \cite{hanhyperattention}, and Polysketchformer \cite{kachampolysketchformer}. 
More recently, \cite{DBLP:journals/corr/abs-2506-12220} provides a new direction of computing large input length transformers with smaller input length transformers.

Practical attempts at improving transformer performance typically turn towards hardware optimization, including techniques such as FlashAttention which optimizes I/O complexity \cite{dao2022flashattention, fu2023simple, sy:24} and speculative decoding which leverages parallelism \cite{leviathan2023fast}.

\section{Preliminaries}

\textbf{Notations.} Let $[n]$ denote the set $\set{1, \dotsc, n}$.
Let $\R_{\geq 0}$ denote the set of non-negative reals and $\R_{> 0}$ denote the set of positive reals.
$\R_{\leq 0}, \R_{< 0}$ are defined analogously.
Our logarithmic function $\log$ is the base $2$ $\log$ unless otherwise specified.

For a ring $R$, let $R[x_{1}, \dotsc ,x_{n}]$ denote the set of polynomials over $n$ variables with coefficients in $R$,
and let $R[[z]] = \{\sum_{n=0}^{\infty}a_nz^n|a_n \in R\}$ denote the set of formal power series over $z$ with coefficients in $R$.

Let $\evec{i}$ denote the $i$-th basis vector with $0$ in every coordinate but $1$ in the $i$-th coordinate.
Let $I_{n}$ denote the identity matrix of dimension $n$.
For any matrix $A$, denote the $(i, j)$-th entry with $A_{ij}$ or $A[i, j]$; we typically use the latter if the matrix name itself has a subscript.
Let $A[i; ]$ denote the $i$-th row of $A$ and $A[; j]$ the $j$-th column of $A$.
For a vector $v \in \R^{n}$, let $\diag(v)$ denote the $n \times n$ matrix with entries $\diag(v)_{ii} = v_{i}$.
For a matrix $A$, let $A^{\T}$ denote its transpose and $A^{-1}$ its inverse.
For any function $f$, let $f(A)$ denote the matrix obtained by applying $f$ to each entry of $A$.
Similarly, let $A \leq B$ denote that $A_{ij} \leq B_{ij}$ for every entry of $A, B$.
Let $A \odot B$ denote the entry-wise product between two matrices.
Let $0_{a \times b}, 1_{a \times b}$ denote the $a \times b$ matrix consisting of all zeros and all ones respectively.
Let $I_{n}$ denote the $n \times n$ identity matrix.

The softmax operator is defined as 
\[
\sm(v) = \frac{1}{\sum_{j=1}^N \exp(v_j)}\cdot (\exp(v_1), \dots, \exp(v_N)) \in \mathbb{R}^N
\] for $v \in \R^N$. We also define the hardmax operator as 
\begin{equation*}
    \hm(v)[i] = \frac{\ind{i \in \imax(v)}}{|\imax(v)|} 
\end{equation*} 
where $\imax(v) = \set{j: v_{j} = \max v}$. Given a matrix $A$, we abuse notation and use $\sm(A)$ and $\hm(A)$ to denote the matrix after we apply $\sm$ and $\hm(A)$, respectively, on each row of $A$.

\subsection{Arithmetic Circuits}

\begin{definition}
    \label{def:arithmetic-circuit}
    An \emph{arithmetic circuit} (\AC) is a sequence of operations $\set{g_1, \dotsc, g_{s}}$ such that each $g_{j} \gets y \circ z$ for $\circ \in \set{+, -, \times, /}$ and $y, z \in \set{g_{i}}_{i < j} \cup \set{x_{i}} \cup \R$ with a subset of gates $\set{g_{i_k}}$ denoted the output gates.
    Given input $x$, denote $\cC(x)$ as the values on the output gates. An arithmetic circuit $\cC$ computes $F: \R^{n} \rightarrow \R^{m}$ if $\cC(x) = F(x)$ for all $x \in \R^{n}$.
\end{definition}

The matrix multiplication exponent $\omega$ is defined as the smallest constant such that the product of two $n$ by $n$ matrices can be computed by an arithmetic circuit of size $O(n^{\omega+\eps})$ for any $\eps > 0$. 

We will also consider a more general class of circuits, which include exponential and logarithmic gates.

\begin{definition}
    \label{def:extended-arithmetic-circuit}
    An \emph{extended arithmetic circuit} (\eAC) computing $F: \R^{n} \rightarrow \R$ is a sequence of operations $\set{g_1, \dotsc, g_{s}}$ such that each $g_{j}$ is one of the following forms: 
    \begin{enumerate}
        \item $g_{j} \gets y \circ z$ for $\circ \in \set{+, -, \times, /}$ and $y, z \in \set{g_{i}}_{i < j} \cup \set{x_{i}} \cup \R$.
        \item $g_{j} \gets \exp(y)$ for $y \in \set{g_{i}}_{i < j} \cup \set{x_{i}} \cup \R$.
        \item $g_{j} \gets \ln(y)$ for $y \in \set{g_{i}}_{i < j} \cup \set{x_{i}} \cup \R_{> 0}$.
    \end{enumerate}
\end{definition}

\subsection{Transformers}
\label{ssec:transformers}

We first define a self-attention head, the core primitive of a transformer. Let $N$ denote the context (input) length, $m$ the embedding (hidden) dimension, and $\din$ the input dimension.
Recall that following practical values, we assume $m = O(N)$.
For an input $X \in \R^{N \times \din}$, we refer to each row of $X$ as a token.

\begin{definition}[Attention]
    \label{def:attn}
    A \emph{self-attention head} is a mapping $f_{Q, K, V}: \R^{N \times \din} \to \R^{N \times m}$ defined by
    \[f_{Q, K, V}(X)= \sm(Q(X) K(X)^\T) V(X) \]
    and parameterized by row-wise
    \emph{query}, \emph{key}, and \emph{value embeddings} $Q, K, V \colon \R^{N\times \din} \to \R^{N \times m}$ (e.g., $Q(X) = (Q(X_1), \dots, Q(X_N))$.
    Let $\attn{m, \din}{N}$ denote the set of all self-attention heads with input dimension $\din$, embedding dimension $m$ and input length $N$.
    When $m = \din$, we denote this as $\attn{m}{N}$. 
\end{definition}

\Cref{fig:attention} illustrates an attention head. 
In practice, and in our constructions, embedding functions typically take the form $Q(X) = X W_{Q}$ (analogously $K(X) = X W_{K}, V(X) = X W_{V}$) for some weight matrices $W_{Q}, W_{K}, W_{V} \in \R^{\din \times m}$.

\begin{figure}
    \centering
    \definecolor{grayq}{gray}{0.85}
\definecolor{grayk}{gray}{0.7}
\definecolor{grayqk}{gray}{0.55}
\definecolor{grayv}{gray}{0.4}

\begin{tikzpicture}[scale=4]

  \draw[line width=1pt] (0,0) rectangle (0.4, 1);
  \node[scale=0.4, label={center:$X$}] (gc4) at (0.2, 0.5) {};
  \draw[decorate,decoration={brace,amplitude=5pt,mirror,raise=1ex}]
  (0,0) -- (0.4,0) node[midway,yshift=-2em]{$\din$};
  \draw[decorate,decoration={brace,amplitude=5pt,raise=1ex}]
  (0,0) -- (0,1) node[midway,xshift=-1.6em]{$N$};

  \draw[line width=1pt, fill=grayq] (0.7, 0) rectangle (1, 1);
  \node[scale=0.4, label={center:$Q(X)$}] (gc4) at (0.85, 0.5) {};
  \draw[line width=1pt, arrows={-latex}] (0.4, 0.5) -- (0.7, 0.5);
  \draw[decorate,decoration={brace,amplitude=5pt,mirror,raise=1ex}]
  (0.7,0) -- (1,0) node[midway,yshift=-2em]{$m$};

  \draw[line width=1pt, fill=grayk] (1.1, 1.1) rectangle (2.1, 1.4);
  \node[scale=0.4, label={center:$K(X)$}] (gc4) at (1.6, 1.25) {};
  \draw[line width=1pt, arrows={-latex}] (0.27, 1) -- (0.27, 1.25) -- (1.1, 1.25);

  \draw[line width=1pt, fill=grayqk] (1.1, 0) rectangle (2.1, 1);
  \node[scale=0.4, label={center:$Q(X)K(X)^{\T}$}] (gc4) at (1.6, 0.5) {};
  \draw[decorate,decoration={brace,amplitude=5pt,mirror,raise=1ex}]
  (1.1,0) -- (2.1,0) node[midway,yshift=-2em]{$N$};

  \draw[line width=1pt, fill=grayv] (2.2, 0) rectangle (2.5, 1);
  \node[scale=0.4, label={center:$V(X)$}] (gc4) at (2.35, 0.5) {};
  \draw[line width=1pt, arrows={-latex}] (0.13, 1) -- (0.13, 1.5) -- (2.35, 1.5) -- (2.35, 1);
  \draw[decorate,decoration={brace,amplitude=5pt,mirror,raise=1ex}]
  (2.2,0) -- (2.5,0) node[midway,yshift=-2em]{$m$};
\end{tikzpicture}
    \caption{An attention head with input $X \in \R^{N \times \din}$, input length $N$ and embedding dimension $m$. $Q(X), K(X), V(X)$ are obtained from $X$ with row-wise embedding maps. Then, the output is computed by $\sm(Q(X)K(X)^{\T})V(X)$.}
    \label{fig:attention}
\end{figure}
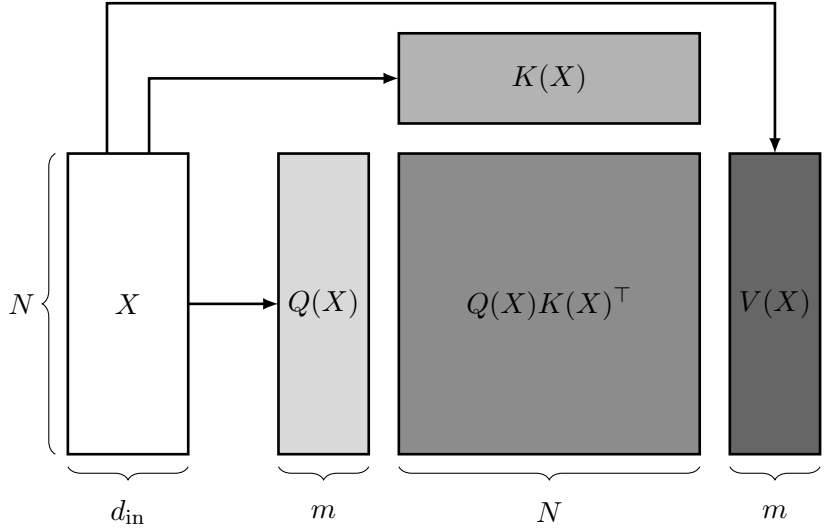

Hardmax attention has also been widely studied in the literature \cite{10.1162/tacl_a_00493, DBLP:conf/icml/Sanford0T24, jerad2025uniquehardattentiontale}, especially in theoretical work. For each query, hardmax attention only allows it to attend to the keys that achieve the maximum inner product (highest correlation).

\begin{definition}[Hardmax Attention]
    \label{def:hardmax-attn}
    A \emph{hardmax-attention head} is a mapping $f^{\hard}_{Q, K, V}: \R^{N \times \din} \to \R^{N \times m}$ defined by
    \[f^{\hard}_{Q, K, V}(X)= \hm(Q(X) K(X)^\T) V(X) \text{.} \]
\end{definition}

In general, we use (self-)attention to mean softmax attention (\Cref{def:attn}) unless otherwise specified. Finally, we define the \emph{multi-layer perceptron} (MLP), which is a shallow neural network. By the universal approximation theorem which states that any continuous function can be approximated by a MLP (with enough neurons), we model MLP as continuous functions on vectors. 

\begin{definition}
    \label{def:mlp}
    An MLP is a continuous function $\Psi: \R^{m_1} \rightarrow \R^{m_2}$ for $m_1, m_2 > 0$.
\end{definition} 

Given a matrix $A$, we abuse notation and denote $\Psi(A)$ as the result of applying $\Psi$ row-wise on matrix $A$. 
We now formally define transformers.

\begin{definition}[Transformer]\label{def:tran}
    A \emph{transformer} is a mapping $T: \R^{N \times \din} \to \R^{N \times \dout}$ specified by attention heads $\set{f_{\ell, h}}_{\ell, h = 1}^{L, H}$ and
    element-wise MLP layers $\Psi_{\ell}: \R^{m} \rightarrow \R^{m}$ for $\ell \in [L]$.
    Upon input $X \in \R^{N \times \din}$, the transformer computes intermediate embeddings $X^{(1)}, \dots, X^{(L)}$ with
    \begin{equation*}
        X^{(\ell)} = \Psi_{\ell}\left(X^{(\ell - 1)} + \sum_{h=1}^{H} f_{\ell,h}\left(X^{(\ell-1)}\right)\right)~~~\forall \ell \in [L],
    \end{equation*}
    and returns $T(X) = \Psi_O(X^{(L)})$ for an output MLP $\Psi_O: \R^m \rightarrow \R^\dout$ (let $X^{(0)} := X$ for notational convenience).  

    A \emph{hardmax transformer} 
    is defined analogously, replacing each attention head with a hardmax-attention head.
\end{definition} 

Note that in this definition, we sum up the result of each attention head, which is considered in many theoretical works (e.g., \cite{DBLP:conf/nips/SanfordHT23, DBLP:conf/icml/Sanford0T24,DBLP:journals/corr/abs-2408-14332}). In the more standard version used in practice, the results of the attention heads are concatenated.
In addition, in sum aggregation, the attention head dimension equals the embedding dimension $m$, while in concatenation aggregation, the attention head dimension equals $m / H$.
As we can easily obtain the sum from the concatenation using a linear MLP, all our lower bounds can be adapted to the concatenation version. See \Cref{app:head-aggregation} for details.

In this paper, we focus on attention mechanisms when proving hardness results, and show lower bounds for transformers without MLPs.
In \Cref{app:mlp-examples}, we show that our results hold also for transformers with standard MLPs found in practice, even when all MLPs take linear time to compute, i.e. given input $x \in \mathbb{R}^m$, computing $\Psi(x)$ takes $O(m)$ time for MLP $\Psi$. 
\Cref{fig:transformer} illustrates an example of a transformer with $3$ heads and $3$ layers.

\begin{figure}
    \centering
    \definecolor{grayAtt}{gray}{0.8}
\definecolor{grayMLP}{gray}{0.6}

\begin{tikzpicture}[scale=3.5]

  \draw[line width=1pt] (0,0) rectangle (0.3, 1);
  \node[scale=0.4, label={center:$X^{(0)}$}] (gc4) at (0.15, 0.5) {};
  \draw[decorate,decoration={brace,amplitude=5pt,mirror,raise=1ex}]
  (0,0) -- (0.3,0) node[midway,yshift=-2em]{$\din$};
  \draw[decorate,decoration={brace,amplitude=5pt,raise=1ex}]
  (0,0) -- (0,1) node[midway,xshift=-1.6em]{$N$};

  \foreach \i in {1,...,3} {
    \tikzmath{\x1 = 0.3 + (\i - 1)*1.1; \y1 = \x1 + 0.3; \z1 = \y1 + 0.2; \c1 = \y1 - 0.1; \d1 = \z1 + 0.1; \e1 = \d1 + 0.2; \f1 = \e1 + 0.3; \g1 = \e1 + 0.15;}
    
    \draw[line width=1pt, fill=grayAtt] (\y1, -0.3) rectangle (\z1, 0.2);
    \draw[line width=1pt, arrows={-latex}] (\x1, 0.4) -- (\c1, 0.4) -- (\c1, -0.05) -- (\y1, -0.05);
    \draw[line width=1pt, arrows={-latex}] (\z1, -0.05) -- (\d1, -0.05) -- (\d1, 0.4) -- (\e1, 0.4);
    \draw[decorate,decoration={brace,amplitude=5pt,mirror,raise=1ex}]
  (\y1,-0.3) -- (\z1,-0.3) node[midway,yshift=-2em]{$m$};
  
    \draw[line width=1pt, fill=grayAtt] (\y1, 0.25) rectangle (\z1, 0.75);
    \draw[line width=1pt, arrows={-latex}] (\x1, 0.5) -- (\y1, 0.5);
    \draw[line width=1pt, arrows={-latex}] (\z1, 0.5) -- (\e1, 0.5);

    \draw[line width=1pt, fill=grayAtt] (\y1, 0.8) rectangle (\z1, 1.3);
    \draw[line width=1pt, arrows={-latex}] (\x1, 0.6) -- (\c1, 0.6) -- (\c1, 1.05) -- (\y1, 1.05);
    \draw[line width=1pt, arrows={-latex}] (\z1, 1.05) -- (\d1, 1.05) -- (\d1, 0.6) -- (\e1, 0.6);

    \draw[line width=1pt] (\e1,0) rectangle (\f1, 1);
    \node[scale=0.4, label={center:$\Psi_{\i}$}] (gc4) at (\g1, 0.5) {};
    \draw[line width=1pt, dotted, arrows={-latex}] (\x1, 0.8) -- (\x1 + 0.1, 0.8) -- (\x1 + 0.1, 1.4) -- (\e1 - 0.1, 1.4) -- (\e1 - 0.1, 0.8) -- (\e1, 0.8);
  }

  \draw[line width=1pt, arrows={-latex}] (3.6, 0.5) -- (3.8, 0.5);
  \draw[line width=1pt] (3.8,0) rectangle (4.1, 1);
  \node[scale=0.4, label={center:$\Psi_{O}$}] (gc4) at (3.95, 0.5) {};
  \draw[decorate,decoration={brace,amplitude=5pt,mirror,raise=1ex}]
  (3.8,0) -- (4.1,0) node[midway,yshift=-2em]{$\dout$};
\end{tikzpicture}
    \caption{A transformer with $H = 3$ and $L = 3$.
    Attention heads are denoted with light gray and MLPs are denoted by $\Psi$.
    Residual connections, denoted with dotted lines, additionally add $X^{(\ell - 1)}$ to $X^{(\ell)}$.
    }
    \label{fig:transformer}
\end{figure}
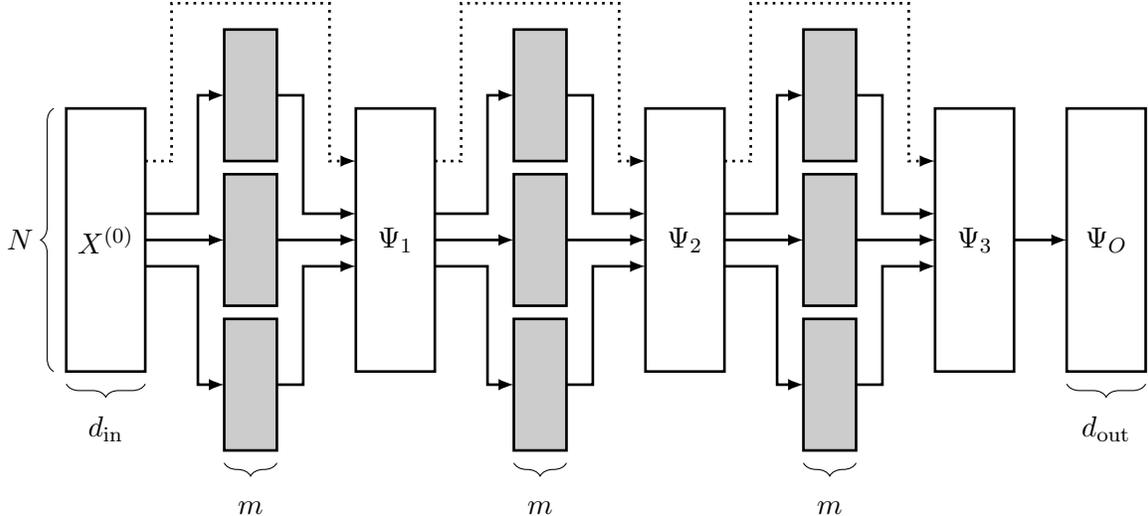

Finally, we are ready to define what it means for an algorithm to compute transformers.

\begin{definition}
    \label{def:alg-compute-tran}
    An algorithm computes transformers with context length $N$, embedding dimension $m$, $L$ layers, $H$ heads in each layer, input dimension $\din$ and output dimension $\dout$ if given any such transformer $T$ and any input $X \in \R^{N \times \din}$, the algorithm outputs $T(X)$.
\end{definition}

For circuits, the definition is a bit more subtle. If we were to attempt to say a circuit $\cC$ computes a transformer, then the circuit must be able to compute all possible embedding and MLP maps, then the circuit must encode many functions at once, which is too strong of a condition to require.
Instead, we aim to give lower bounds against circuits meeting a much weaker requirement: we show lower bounds against circuits $\cC$ that compute fixed transformers.

Although we consider fixed transformers, we remark that the embedding map may be parameterized.
This corresponds in practice to the learned weights of a model.
Formally, consider an attention head $f_{\ell, h}$ with embedding maps $Q_{h, \ell}, K_{h, \ell}, V_{h, \ell}$.
Then, we require that $(Q_{h, \ell})_{i} = g(w_{Q_{h, \ell}}, X_{i})$ where $g$ is an function computable by an \eAC, $w_{Q_{h, \ell}}$ denotes the weights associated to this map, and $X_{i}$ denotes the $i$-th row of the input to the attention head. 
In our constructions, we require that $g$ is an \eAC of size $O(m)$.
Let $W$ denote the concatenation of all the weights of the transformer, and for any \eAC transformer $T$, let $T_{W}(X)$ be the computation of $T$ with weights $W$ on input $X$.
Finally, we define what it means for a circuit to compute a transformer.

\begin{definition}
    \label{def:circuit-compute-tran}
    Let $T$ be a transformer with context length $N$, embedding dimension $m$, $L$ layers, $H$ heads in each layer, input dimension $\din$ and output dimension $\dout$.
    Suppose $T$ has weights $W$ and input $X$.
    A circuit $\cC$ computes $T$ if for any input $X \in \R^{N \times \din}$ and any weights $W$, $\cC(X, W) = T_{W}(X)$.
\end{definition}

To obtain meaningful lower bounds against small \eAC computing transformers, we show that transformers can compute functions that are hard for extended arithmetic circuits.
Of course, to have a meaningful lower bound, we note that the transformer itself must also be computable with an \eAC.

\paragraph{Softmax Simulates Hardmax.}

We introduce a useful primitive for Transformer computation and show that softmax attention can simulate hardmax attention when there is a sufficient gap between the maximum value and the second largest value.
This is similar to a result of \cite{DBLP:conf/icml/Sanford0T24}.

\begin{restatable}{lemma}{SoftmaxSimHardmax}
    \label{lem:softmax-sim-hardmax}
    Let $f \in \attn{m}N$ be a self-attention unit with embedding functions $Q, K, V$ such that
    \begin{enumerate}
        \item $V(X) \in \set{0, 1, 2}^{n \times m}$, and
        \item for every $X \in \R^{N \times m}$ and $i \in [N]$:
        \[
        A(X)_{i, i'} \leq \max_{i''} A(X)_{i, i''} - 1, \ \forall i' \not\in \imax(A(X)_i),
        \] where $A(X) = Q(X) K(X)^\T$.
    \end{enumerate} Then there exists an attention unit $f' \in \attn{m}N$ satisfying 
    \begin{equation*}
        \norm{f'(X) - \hm(A(X)) V(X)}_{\infty} \leq \frac{1}{\poly(N)}
    \end{equation*}
    for an arbitrarily small inverse polynomial.
\end{restatable}

The idea behind the simulation is simple: by scaling each query vector in $Q$ by a sufficiently large constant $c = \Theta(\log N)$, we ensure that any probability assigned to non-maximum indices by the softmax operation is exponentially small with respect to $c$, i.e. less than $\frac{1}{\poly(N)}$ for any polynomial. We leave the proof to \Cref{app:omitted-proofs}.

\subsection{Fine-Grained Complexity}

We state our major hypotheses in fine-grained complexity theory in this section. See \Cref{sec: intro} for more discussions on their applications.

\begin{hypothesis}[Strong Exponential Time Hypothesis ($\SETH$)]
For any $\varepsilon>0$, there exists a positive integer $k$ such that $\kSAT$ requires $\Omega(2^{(1-\varepsilon)n})$ time, where $n$ is the number of variables in the CNF.
\end{hypothesis}

We say that $k$ vectors $v_1, \dots, v_{k} \in \set{0, 1}^{d}$ are orthogonal if for all $j \in [d]$, $v_{\ell}[j] = 0$ for some $\ell \in [k]$. In other words, their $k$-wise inner product is $0$. The $\kOV$ problem is then defined as follows: given $A_{1}, \dots, A_{k} \subset \set{0, 1}^{d}$ each of size $n$, determine if there exists $v_{1} \in A_{1}, \dots, v_{k} \in A_{k}$ such that $v_{1}, \dots, v_{k}$ are orthogonal.
The trivial algorithm of checking all $k$-tuples of vectors requires $O(n^{k} d)$ time.
The $\kOV$ Hypothesis states that for $d = \Theta(\log n)$, this is essentially optimal: i.e. any algorithm requires $n^{k - o(1)}$ time.

\begin{hypothesis}[$\kOV$ Hypothesis]
    \label{hyp:k-ov}
    Computing $\kOV$ on sets of $n$ vectors with dimension $\Theta(\log n)$ requires $n^{k - o(1)}$ time. 
\end{hypothesis}

The $\kOV$ Hypothesis follows from the popular Strong Exponential Time Hypothesis (\SETH)~\cite{impagliazzo2001seth} which states that satisfiability over $n$ variables requires $2^{(1 - o(1))n}$ time \cite{williams2004ovc}.
Our lower bound is based on the $\kOVf{3}$ Hypothesis, a special case of $\kOV$ for $k = 3$.

The unbalanced $\kOV$ problem is defined as follows: given $k$ sets of vectors $A_{1}, \dots, A_{k}$ of size $S_{1}, \dots, S_{k}$, determine if there exist $v_{1} \in A_{1}, \dots, v_{k} \in A_{k}$ such that $v_{1}, \dots, v_{k}$ is orthogonal.
The trivial algorithm of checking all $k$-tuples requires $O(S_1S_2\dots S_{k}d)$ time.
Again, the Unbalanced $\kOV$ Hypothesis states that this is essentially optimal and it is well-known that the Unbalanced $\kOV$ Hypothesis is equivalent to the $\kOV$ Hypothesis. 

\begin{hypothesis}[Unbalanced $\kOV$ Hypothesis]
    \label{hyp:unbalanced-k-ov}
    $\kOV$ on sets $S_{1}, \dots, S_{k}$ of size $n^{s_{1}}, \dots, n^{s_{k}}$ with dimension $d = \Theta(\log n)$ requires $n^{s_{1} + \dots + s_{k} - o(1)}$ time. 
\end{hypothesis}

It is clear that the Unbalanced $\kOV$ Hypothesis implies the $\kOV$ Hypothesis, by setting $s_{1} = \dots = s_{k} = 1$.
The following lemma shows the converse direction, and we include a proof in \Cref{app:omitted-proofs} for completeness. 

\begin{restatable}{lemma}{UnbalancedKOVEquiv}
    \label{lemma:unbalanced-k-OV}
    Let $0 < s_{2} \dotsc, s_{k} \leq 1$. Let $A_1, \dots A_{k} \subset \set{0, 1}^{d}$ for $d = \Theta(\log n)$ with $|A_{1}| = n$ and $|A_{i}| = n^{s_{i}}$ for all $i > 1$. Then, any algorithm computing $\kOV$ on $A_1, \dotsc, A_{k}$ requires time $n^{s_1 + \dotsc + s_{k} - o(1)}$ under the $\kOV$ Hypothesis. In particular, the Unbalanced $\kOV$ Hypothesis holds under the $\kOV$ Hypothesis.
\end{restatable}

\section{Lower Bounds for Small Embedding Dimension}
We begin with lower bounds for the small embedding dimension setting (i.e. $m = N^{o(1)}$).
Our lower bound is based on the $\kOV$ Hypothesis (or $\SETH$).

First, we show that under the $\kOV$ Hypothesis, when the embedding dimension is small (e.g. $m = \Theta(\log n)$), there is no algorithm that computes transformers  significantly more efficiently than the naive algorithm that computes each head separately.

\begin{restatable}[Formal \Cref{thm:small-embedding}]{theorem}{SmallEmbeddingLBFormal}
    \label{thm:small-embedding-formal}
    Let $m = \Theta(\log N)$ and $L, H = \poly(N)$.
    Any algorithm computing transformers with input length $N$, embedding dimension $m$, $L$ layers, $H$ heads in each layer, and input and output dimension $m$ up to entry-wise additive error $\frac{1}{10N}$ requires $LHN^{2 - o(1)}$ time under the $\kOVf{3}$ Hypothesis.
    The lower bound holds even if the transformer has no MLPs.
\end{restatable}

In the above statement, we require $L, H = \poly(N)$ since in $\kOVf{3}$ the sizes of the three sets must be polynomially related.
We remark that this is the interesting regime of the problem since when $L, H = N^{o(1)}$, we obtain optimal lower bounds directly from the hardness of a single attention head.

\begin{lemma}
    \label{lemma:transformer-compute-k-OV}
    Suppose there exists an algorithm $\mathcal{A}$ that computes any transformer with input length $N+1$, embedding dimension $2d+2$, $L$ layers, $H$ heads in each layer, input and output dimension $2d+2$, up to entry-wise additive error $\frac{1}{10N}$ in time $B$, then
    there exists an algorithm that solves unbalanced $\kOVf{3}$ with dimension $d$ and size $N,N,LH$ in time $\tO{LHd+B}$.   
\end{lemma}

First, we observe that \Cref{thm:small-embedding-formal} follows from \Cref{lemma:transformer-compute-k-OV}.

\begin{proof}[Proof of \Cref{thm:small-embedding-formal} assuming \Cref{lemma:transformer-compute-k-OV}]
    Suppose for contradiction that there is an algorithm computing any transformer with input length $N+1$, embedding dimension $m = 2d+2$, $L$ layers, $H$ heads in each layer, input and output dimension $m$ up to entry-wise additive error $\frac{1}{10N}$ in time $\bigO{LHN^{2 - \eps}}$ for some $\eps > 0$.
    Then, consider an unbalanced $\kOVf{3}$ instance with dimension $d = \Theta(\log N)$ and size $N, N, LH$.
    \Cref{lemma:transformer-compute-k-OV} implies that there is a $\tO{LHd + LHN^{2-\varepsilon}}$ algorithm for this problem, which contradicts the $\kOVf{3}$ Hypothesis.
\end{proof}

We now prove \Cref{lemma:transformer-compute-k-OV} to complete the proof of \Cref{thm:small-embedding-formal}.

\begin{proof}[Proof of \Cref{lemma:transformer-compute-k-OV}]
    Consider a $\kOVf{3}$ instance $(A, B, C)$ with $|A| = |B| = N$ and $|C| = LH$.
    We construct the following transformer $T$ with $m = \din = \dout = 2d + 2$.
    Suppose we have an algorithm computing $T$ up to entry-wise additive error $\frac{1}{10N}$.
    Let the input $X \in \R^{(N + 1) \times \din}$ be
    \begin{align*}
        X &= \begin{pmatrix}
            a_1^{\T} & b_1^{\T} & 1 & 0 \\
            a_2^{\T} & b_2^{\T} & 1 & 0  \\
            \vdots & \vdots & \vdots & \vdots \\
            a_{N}^{\T} & b_{N}^{\T} & 1 & 0 \\
            0^{\T} & 0^{\T} & 2 & 0
        \end{pmatrix},
    \end{align*}
    where $a_i, b_i$ are vectors in $A, B$. Additionally, we index $C = \set{c_{h, \ell}: 1 \leq h \leq H, 1 \leq \ell \leq L}$ and define
    \begin{align*}
        Q_{h, \ell}(X) = - \begin{pmatrix}
            a_1^{\T} & 0_{1 \times (d + 2)} \\
            a_2^{\T} & 0_{1 \times (d + 2)} \\
            \vdots & \vdots \\
            a_{N}^{\T} & 0_{1 \times (d + 2)} \\
            0_{1 \times d} & 0_{1 \times (d + 2)}
        \end{pmatrix},K_{h, \ell}(X) = \begin{pmatrix}
            b_1^{\T} \odot c_{h, \ell}^{\T} & 0_{1 \times (d + 2)} \\
            b_2^{\T} \odot c_{h, \ell}^{\T} & 0_{1 \times (d + 2)} \\
            \vdots & \vdots \\
            b_{N}^{\T} \odot c_{h, \ell}^{\T} & 0_{1 \times (d + 2)} \\
            0_{1 \times d} & 0_{1 \times (d + 2)}
        \end{pmatrix}, V_{h, \ell}(X) = \begin{pmatrix}
            0_{1 \times (2d + 1)} & 1 \\
            \vdots  & \vdots\\
            0_{1 \times (2d + 1)} & 1 \\
            0_{1 \times (2d + 1)} & 2 
        \end{pmatrix}
    \end{align*}
    such that $Q_{h,\ell}(X) \in \mathbb{R}^{(N+1) \times m}, K_{h,\ell}(X) \in \mathbb{R}^{(N+1) \times m}, V_{h,\ell} \in \mathbb{R}^{(N+1)\times m}$ where $b \odot c \in \R^{d}$ is the element-wise product of two vectors $a, b$.
    
    We claim that creating the transformer requires $O(d H L)$ time since each attention head can be created in time $O(d)$ time (i.e. the embedding maps can be specified in $O(d)$ time).
    In addition, we can encode $Q_{h, \ell}(X), K_{h, \ell}(X)$ as desired. 
    In particular, $Q_{h, \ell}, K_{h, \ell}: \R^{m} \rightarrow \R^{m}$ apply the following maps to each row:
    \begin{align*}
        Q_{h, \ell}&: v \mapsto - (v_{1}, \dots, v_{d}, 0, \dots, 0)\text{,} 
        \\ K_{h, \ell}&: v \mapsto (v_{d + 1} \cdot c_{h, \ell}[1],\dots, v_{2d} \cdot c_{h, \ell}[d], 0, \dots, 0) \text{,}
        \\ V_{h, \ell}&: v \mapsto (0, \dots, 0, 0, v_{2d + 1}) \text{.}
    \end{align*}
    Observe that the above maps can be expressed using weight matrices in $\R^{m \times m}$:
    \begin{align*}
        W_{Q, h, \ell} &= \begin{pmatrix}
            - I_{d \times d} & 0_{d \times (d + 2)} \\
            0_{(d + 2) \times d} & 0_{(d + 2) \times (d + 2)}
        \end{pmatrix}\text{,} \\
        W_{K, h, \ell} &= \begin{pmatrix}
            0_{d \times d} & 0_{d \times (d + 2)} \\
            \diag(c_{h, \ell}) & 0_{d \times (d + 2)} \\
            0_{2 \times d} & 0_{2 \times (d + 2)}
        \end{pmatrix}\text{,}\\
        W_{V, h, \ell} &= \begin{pmatrix}
            0_{2d \times (2d + 1)} & 0_{2d \times 1} \\
            0_{1 \times (2d + 1)} & 1 \\
            0_{1 \times (2d + 1)} & 0
        \end{pmatrix}
    \end{align*}
    so that $Q_{h, \ell}(X) = X W_{Q, h, \ell}$ and analogously for $K_{h, \ell}, V_{h, \ell}$.
     
    Firstly, we observe that the first $2d+1$ columns will never change throughout the computation. 
    We remind the readers that $X^{(\ell)}$ is the output of $\ell$-th transformer layer.
    
    \begin{claim}
        \label{clm:3-ov-2-col}
        For all $\ell$, the first $2d + 1$ columns of $X^{(\ell)}$ are the first $2d + 1$ columns of $X$. 
    \end{claim}

    \begin{proof}
        The claim follows as the first $2d + 1$ columns of $V_{h, \ell}(X)$ are $0$ for all $h, \ell$. 
        As a result, the output of each head consists of $0$ in the first $2d + 1$ columns.
        Then, since $X^{(\ell)} = X^{(\ell - 1)} + \sum_{h} f_{h, \ell}(X^{(\ell - 1)})$ where $f_{h, \ell}$ denotes the $h$-th attention head in the $\ell$-th layer, we have that the first $2d + 1$ columns of $X^{(\ell)}$ are the first $2d + 1$ columns of $X^{(\ell - 1)}$.
        By induction, the first $2d + 1$ columns of $X^{(\ell)}$ are the first $2d + 1$ columns of $X^{(0)} = X$.
    \end{proof}
    
Secondly, we claim that the given hardmax transformer computes the $\kOVf{3}$ instance. Notice that this suffices because of \Cref{lem:softmax-sim-hardmax}, as an algorithm computing any softmax transformer can approximate any hardmax transformer with the same dimension.
    Consider the $(h, \ell)$-th attention head.
    Note that
    \begin{equation*}
        (Q_{h, \ell}(X) K_{h, \ell}(X)^{\T})_{i, j} = \begin{cases}
            0 & \textup{if } \max(i, j) = N + 1 \\
            - a_{i} \cdot b_{j} \cdot c_{h, \ell} & \textup{otherwise}.
        \end{cases}
    \end{equation*}
    
    For $1 \leq i \leq N$, the maximum entry in the $i$-th row is $0$, since $(Q_{h, \ell}(X)K_{h, \ell}(X)^{\T})_{i, N+1} = 0$.
    Recalling that $\imax(A(X)_{i})$ denotes the set of indices attaining maximum value in the $i$-th row, we have that $|\imax(A(X)_{i})| \geq 2$ if $(a_{i}, c_{h, \ell})$ participate in an orthogonal triple and $|\imax(A(X)_{i})| = 1$ otherwise. 
    In particular, let $t_{i} := t_{i, (h, \ell)}$ denote the number of $j$ such that $a_{i} \cdot b_{j} \cdot c_{h, \ell} = 0$.
    Then we have
    \begin{equation*}
        \hm(Q_{h, \ell}(X) K_{h, \ell}(X)^{\T})V_{h, \ell}(X) = \begin{pmatrix}
            0 & \cdots & 0 & \frac{t_{1} + 2}{t_{1} + 1} \\
            0 & \cdots & 0 & \frac{t_{2} + 2}{t_{2} + 1} \\
            \vdots & \ddots & \vdots & \vdots \\
            0 & \cdots & 0 & \frac{t_{N} + 2}{t_{N} + 1} \\
            0 & \cdots & 0 & \frac{N + 2}{N + 1} 
        \end{pmatrix}
    \end{equation*}
    In particular, for every $i \le N$, the $i$-th entry in the last column is $2$ if $t_{i} = 0$ and at most $\frac{3}{2}$ otherwise.
    \Cref{lem:softmax-sim-hardmax} then implies that there is a transformer such that each attention head computes an output $Y_{h, \ell}$ that approximates the hardmax attention head with embeddings $Q_{h, \ell}, K_{h, \ell}, V_{h, \ell}$ well.
    In particular, the output of the $h$-th head in the $\ell$-th layer, denoted $Y_{h, \ell}$, satisfies
    \begin{equation}
        \label{eq:precision-3-ov}
        \norm{Y_{h, \ell} - \hm(Q_{h, \ell}(X) K_{h, \ell}(X)^{\T})V_{h, \ell}(X)}_{\infty} \leq \frac{1}{10 NHL}
    \end{equation}
    for some inverse polynomial in $L, H, N$.
    
    Let $Z_{h, \ell}$ denote the last column of $\hm(Q_{h, \ell}(X) K_{h, \ell}(X)^{\T})V_{h, \ell}(X)$.
    If the $\kOVf{3}$ instance is a no-instance (i.e. there is no orthogonal triple), then every $Z_{h, \ell}[i] = 2$, i.e.
    \begin{equation*}
         \sum_{h, \ell} \sum_{i = 1}^{N} Z_{h, \ell}[i] = 2NHL \text{.}
    \end{equation*}
    On the other hand, if the $\kOVf{3}$ is a yes-instance (i.e. there is an orthogonal triple), then there exists some $i, (h, \ell)$ where $Z_{h, \ell}[i] \leq \frac{3}{2}$ i.e. 
    \begin{equation*}
         \sum_{h, \ell} \sum_{i = 1}^{N} Z_{h, \ell}[i] \leq 2NHL - \frac{1}{2}\text{.}
    \end{equation*}

    Let us then consider the output $X^{(L)}$, in particular focusing on the last column, since the first $2d + 1$ columns are exactly those of $X = X^{(0)}$ by \Cref{clm:3-ov-2-col}.
    Observe that our transformer outputs for all $1 \leq i \leq N$
    \begin{equation*}
        X^{(L)}[i, m] = X^{(0)}[i, m] + \sum_{h, \ell} Y_{h, \ell}[i, m] = \sum_{h, \ell} Y_{h, \ell}[i, m] \text{.}
    \end{equation*}
    Since our algorithm computes the given transformer up to entry-wise additive error $\frac{1}{10N}$, we in fact obtain in the last column a vector $\hat{X}$ such that $|\hat{X}_{j} - X^{(L)}[i, m]| \leq \frac{1}{10N}$ for all $i \in [N]$.
    
    Thus, from the computed output of the transformer, we can compute in $O(N)$ time $\sum_{j = 1}^{N} \hat{X}[j]$ such that
    \begin{align*}
        \left| \sum_{i = 1}^{N} \hat{X}[i] - \sum_{h, \ell} \sum_{i = 1}^{N} Z_{h, \ell}[i] \right| &\leq \left| \sum_{i = 1}^{N} \hat{X}[i] - \sum_{i = 1}^{N} X^{(L)}[i, m] \right| \\
        &\quad\quad + \left| \sum_{i = 1}^{N} X^{(L)}[i, m] - \sum_{h, \ell} Y_{h, \ell}[i, m] \right| \\
        &\quad\quad + \left| \sum_{h, \ell} Y_{h, \ell}[i, m] - \sum_{h, \ell} Z_{h, \ell}[i] \right| \\
        &\leq 0.2 \text{.}
    \end{align*}
    Here, we apply the triangle inequality, apply the approximation error of our algorithm for the first term, our equation above for the second term and \eqref{eq:precision-3-ov} for the third term.
    Note that from the output of the transformer is a $(N + 1) \times m$ matrix, but we sum only the first $N$ entries of the last column.
    
    Thus, the sum $\sum_{j = 1}^{N} \hat{X}[j]$ is at least $2NHL - 0.2$ if there is no orthogonal triple and at most $2NHL - 0.3$ otherwise.
    Therefore, we can compute the output of the $\kOVf{3}$ instance.
    
\end{proof}

We conclude with a simple conditional lower bound which suggests that efficient transformers require large size to compute $\kOVf{3}$, suggesting that small size transformers cannot handle some instances of $3$-wise interaction.

\begin{theorem}
    Under the $\kOVf{3}$ Hypothesis, any transformer with linear time embedding and MLP maps computing $d$-dimensional $\kOVf{3}$ requires $LHm = N^{1 - o(1)}$.
\end{theorem}

\begin{proof}
    Suppose there is a transformer $T$ with embedding dimension $m = \Theta(\log N)$, $L$ layers, $H$ heads in each layer, input dimension $m$, output dimension $m$ that computes $\kOVf{3}$ with dimension $m$.
    Recall that in an efficient transformer, all embedding and MLP maps require $\tO{m}$ time, so that each attention head is computed in $O(Nm + N^2 m) = O(N^2 m)$ time.
    Following the standard algorithm, this transformer can be computed in time $\tO{N^2 m H L}$.
    In particular, under the $\kOVf{3}$ Hypothesis, this requires $L H m = N^{1 - o(1)}$.
\end{proof}

\section{Tools for Algebraic Circuits}

In this section, we establish some tools for (extended) arithmetic circuits which will be used in our lower bound proofs for large embedding dimensions. 

While \eAC is a more powerful class in general (e.g. \AC cannot compute $\exp(x)$), we show that the additional gates are not necessary for computing matrix multiplication.
More generally, we argue that we can in fact use standard arithmetic circuits to simulate an \eAC that computes low degree functions.
For our application, we give the result for quadratic functions. 
The proof generalizes that of~\cite{strassen1973vermeidung} which states that division gates are unnecessary for quadratic functions.

\begin{theorem}
    \label{thm:ac-simulates-eac}
    Let $F = \{F_{1}, \dotsc, F_{m}\}$ where $F_{i} = \sum_{k,\ell=1}^{n} t_{ik\ell}x_{k}x_{\ell}$.
    Suppose there is an \eAC $\cC$ of size $s$ computing $F$.
    Then there is an \AC of size $O(s)$ computing $F$.
\end{theorem}

\begin{proof}

    Let $\cC$ be an optimal \eAC for computing $F$, and let $W = (w_{1}, \dotsc, w_{s})$ be an ordering of all gates in $\cC$ where earlier gates are computed first.
    Without loss of generality, viewing each $w_{i}(x_{1}, \dots, x_{n})$ as a function, we can assume $w_{i} \neq 0$ for any $i$.
    Furthermore, since $\cC$ computes $F$ for all inputs $x \in \R^{n}$, $\cC$ is well defined for all $x$.
    In particular, if $w_{i}$ is ever provided as input to a $\ln$ gate, we have that $w_{i}(x_{1}, \dots, x_{n}) > 0$ for all $x$.
    Similarly, if $w_{i}$ is ever provided as a denominator to a $/$ gate, we have that $w_{i}(x_{1}, \dots, x_{n}) \neq 0$ for all $x$.
    By substituting $\tilde{x}_{i} \gets x_{i} z$, we can view each gate $w_{j}$ as a formal power series $\tilde{w}_{j} \in \R(x_1, \dots, x_{n})[[z]]$ with coefficients that are functions on $x_{1}, \dots, x_{n}$.
    We can do this as all the formal power series are  well-defined when $\cC$ is well-defined under all the inputs $x \in \R^n$. 
    First, the sum, difference, and product are well defined for any two formal power series, and exponentiation of a formal power series is well-defined (we will show the exact formulation later). 
    For a function $f$ corresponding to some gate, and its formal power series $\tilde{f}$, 
    observe that $\ln(\tilde{f})$ is well defined as long as $\tilde{f}$ has positive constant coefficient (see exact formulation later). 
    The constant coefficient of $\tilde{f}$ is $f(0, \dots, 0)$, and our assumption that inputs to $\ln$ gates are positive on all inputs ensures that $f(0, \dots, 0) > 0$, which implies the constant coefficient of $\tilde{f}$  as needed. 
    For a division gate $w_{j} = g/h$, our assumption that any gate input to the denominator of a division gates is non-zero for all inputs ensures that the constant coefficient of $\tilde{h}$ is $h(0, \dots, 0) \neq 0$.
    The following fact then guarantees that the formal power series of $w_{j}$ is well defined.

    \begin{fact}
        \label{fact:invertible-power-series}
        Let $R$ be a ring. 
        A power series $\tilde{h}(z) = \sum_{i} a_{i} z^{i} \in R[[z]]$ is invertible if and only if $a_{0} \neq 0$. 
        Furthermore, the inverse is given by
        \begin{equation*}
            \frac{1}{\tilde{h}(z)} = \frac{1}{a_{0}} \left( 1 + q + q^2 + q^3 + \dotsc \right)
        \end{equation*}
        where $q = - \sum_{i = 1}^{\infty} \frac{a_{i}}{a_{0}} z^{i}$.
    \end{fact}

    Viewing the gates as formal power series, we argue that the desired output can easily be computed from the low order coefficients of the power series of the relevant gates. Specifically, since all $F_i$ have degree at most $2$, we can add a new variable $z$ to all the input $x_i$ and only look at the coefficients for $z^0 ,z^1, z^2$. Suppose $w_{j}$ is an output gate $F_{i}$.
    Then, we have 
    \begin{equation*}
        F_i = w_{j}(x_{1}, \dots, x_{n}) = \sum_{k, \ell = 1}^{n} t_{i k \ell} x_{k} x_{\ell} \text{.} 
    \end{equation*}
    Notice that 
    \begin{equation*}
        w_{j}(x_{1} z, \dots, x_{n} z) = \sum_{k, \ell = 1}^{n} t_{i k \ell} x_{k} x_{\ell} z^{2} \text{,} 
    \end{equation*}
    whose formal power series is simply itself. 
    Hence, we can compute $F_i$ from the coefficient of $z^2$ in the formal power series $\tilde{w}_j$.

    We have shown that from the second coefficient of the power series corresponding to the output gates we can compute $F$ efficiently.
    While we will not be able to compute the full power series of each gate as there could be infinitely many terms, we describe how to compute the low-order coefficients, i.e. the coefficients of $1, z, z^2$, for all $\tilde{w}_{i}$.
    We proceed by induction.
    We denote $\tilde{w}_{i} = \sum_{j\ge 0} c_{j} z^{j}$ and compute $c_{0}, c_{1}, c_{2}$.
    For the input gates, $w_{i}(x_{1}, \dots, x_{n}) = x_{j}$ for some $x_{j}$ so that $\tilde{w}_{i} = x_{j} z$.
    In particular, $c_{0} = c_{2} = 0$ and $c_{1} = x_{j}$.
    This requires $O(1)$ gates.
 
    {\bf Case 1: Addition and Subtraction. }
    Suppose $\tilde{w}_{i} = \tilde{g}_{i} + \tilde{h}_{i}$ where $\tilde{g}_{i}, \tilde{h}_{i}$ are power series with
    \begin{equation*}
        \tilde{g}_{i} = \sum_{j = 0}^{\infty} u_{j} z^{j} \quad, \quad \tilde{h}_{i} = \sum_{j = 0}^{\infty} v_{j} z_{j} \text{.}
    \end{equation*}
    The first three coefficients can easily be combined via addition and subtraction i.e. $c_{j} = u_{j} \pm v_{j}$ for all $j$.
    Again, this requires $O(1)$ gates.

    {\bf Case 2: Multiplication. }
    Suppose $\tilde{w}_{i} = \tilde{g}_{i} \cdot \tilde{h}_{i}$ where $\tilde{g}_{i} = \sum_{j = 0}^{\infty} u_{j} z^{j}$ and $\tilde{h}_{i} = \sum_{j = 0}^{\infty} v_{j} z^{j}$.
    Then, if $\tilde{w}_{i} = \sum_{j} c_{j} z^{j}$, we have $c_{0} = u_{0} v_{0}$, $c_{1} = u_{0} v_{1} + u_{1} v_{0}$, and $c_{2} = u_{2} v_{0} + u_{1}v_{1} + u_{0}v_{2}$.
    This requires $O(1)$ gates.

    {\bf Case 3: Division. }
    Suppose $\tilde{w}_{i} = \tilde{g}_{i}/\tilde{h}_{i}$ where $\tilde{g}_{i} = \sum_{j = 0}^{\infty} u_{j} z^{j}$ and $\tilde{h}_{i} = \sum_{j = 0}^{\infty} v_{j} z^{j}$.
    Since $v_{0} \neq 0$, we obtain the inverse via \Cref{fact:invertible-power-series} as
    \begin{equation*}
        \frac{1}{\tilde{h}_{i}} = \frac{1}{v_{0}} \left( 1 +q +  q^2 + q^3 + \dotsc \right)
    \end{equation*}
    with $q = - \sum_{j = 1}^{\infty} \frac{v_{j}}{v_{0}} z^{j}$.
    In particular, we have
    \begin{align*}
        \frac{\tilde{g}_{i}}{\tilde{h}_{i}} &= \frac{1}{v_{0}}(u_{0} + u_{1} z + u_{2} z^{2} + \dots) \cdot \left(1 + \left( - \frac{v_{1}}{v_{0}} z - \frac{v_{2}}{v_{0}} z^2 - \dots \right) + \left( - \frac{v_{1}}{v_{0}} z -  \dots\right)^2 + \dots \right) \\
        &= (u_{0}' + u_{1}' z + u_{2}' z^{2} + u_{3}' z^{3} + \dots) \cdot \left(1 + \left( - v_{1}' z - v_{2}' z^2 - \dots \right) + \left( - v_{1}' z -  \dots\right)^2 + \dots \right)
    \end{align*}
    where we set $u_{j}' = \frac{u_{j}}{v_{0}}$ and $v_{j}' = \frac{v_{j}}{v_{0}}$.
    Finally, we have
    \begin{align*}
        c_{0} &= u_{0}' \\
        c_{1} &= u_{1}' - u_{0}' v_{1}' \\
        c_{2} &= u_{2}' - u_{1}' v_{1}' + u_{0}' (v_{1}'^2 - v_{2}') \text{.}
    \end{align*}
    This requires $O(1)$ gates.

    {\bf Case 4: Exponent. }
    Suppose $\tilde{w}_{i} = \exp(\tilde{g}_{i})$ where $\tilde{g}_{i} = \sum_{j=0}^\infty u_{j}z^j$. 
    Then, using the formal power series of $\exp(x) = \sum_{k = 0}^{\infty} \frac{x^{k}}{k!}$, we obtain
    \begin{align*}
        \exp(\tilde{g}_{i}) &= \exp \left( \sum_{j \ge 0} u_{j}z^j \right) \\
        &= \exp(u_{0}) \exp \left(\sum_{j \ge 1} u_{j}z^j \right) \\
        &= \exp(u_{0}) \left(
        \sum_{k \ge 0} \frac{1}{k!} \left( \sum_{j \ge 1} u_{j}z^j\right)^k\right).
    \end{align*}
    Note that for every fixed $j \ge 0$, there are finitely many terms in the above sum that contribute to the coefficient of $z^j$, so the formal power series is well-defined. In particular, 
    \begin{align*}
        c_{0} &= e^{u_{0}} \\
        c_{1} &= e^{u_{0}} u_{1} \\
        c_{2} &= e^{u_{0}} \left( u_{2} + \frac{u_{1}^2}{2} \right) \text{.}
    \end{align*}
    This requires $O(1)$ gates. 
    Even though the formulas use $\exp(u_0)$, it is a constant and can be implemented in an \AC. 

    {\bf Case 5: Logarithm. }
    Suppose $\tilde{w}_{i} = \ln(\tilde{g}_{i})$ where $\tilde{g}_{i} = \sum_{j \ge 0} u_j z^j$.  
    Then applying the power series identity $\ln(1 + x) = \sum_{k=1}^{\infty} (-1)^{k + 1} \frac{x^{k}}{k}$, we have
    \begin{align*}
        \ln(\tilde{g}_{i}) &= \ln u_{0} + \ln \left( 1 + \sum_{j \ge 1} \frac{u_{j}}{u_{0}} z^j \right) \\
        &= \ln u_{0} + \sum_{k \ge 1} \frac{(-1)^{k+1}}{k} \cdot \left(\sum_{j \ge 1} \frac{u_{j}}{u_{0}} z^j\right)^k
    \end{align*}
    Again, for every fixed $j \ge 0$, there are finitely many terms in the above sum that contribute to the coefficient of $z^j$, so the formal power series is well-defined. 
    In particular, 
    \begin{align*}
        c_{0} &= \ln u_{0}\\
        c_{1} &= \frac{u_{1}}{u_0}\\
        c_{2} &= \frac{u_{2}}{u_0} - \frac{1}{2} \left(\frac{u_{1}}{u_0}\right)^2 \text{.}
    \end{align*}
    This requires $O(1)$ gates. Similarly as the $\exp$ case, we only use $\ln$ on a constant, so it can be implemented in the \AC model. 
    
    Overall, our circuit requires $O(s)$ gates to compute all the relevant coefficients.
\end{proof}

The next tool we require is the Baur-Strassen Theorem, which states that given an \AC that computes a function, we can compute all the partial derivatives with a similar sized circuit.

\begin{theorem}[\cite{DBLP:journals/tcs/BaurS83}]
    \label{thm:baur-strassen}
    Let $F: \R^{n} \rightarrow \R$ be a function and $F' = (\pderiv{F}{x_{1}}, \dots, \pderiv{F}{x_{n}})$.
    If there is an \AC of size $s$ computing $F$, there is an \AC of size $O(s)$ computing $F'$.
\end{theorem}

After Baur and Strassen \cite{DBLP:journals/tcs/BaurS83} obtained this result, Morgenstern \cite{DBLP:journals/sigact/Morgenstern85} showed a simplified constructive proof for the Baur-Strassen theorem. 
Following the Morgenstern's proof strategy, we generalize the Baur-Strassen theorem to \eAC. 

\begin{theorem}[Extended Baur-Strassen Theorem]
    \label{thm:baur-strassen-extend}
    Let $F: \R^{n} \rightarrow \R$ be a function and let $F' = \set{\pderiv{F}{x_{1}}, \dotsc, \pderiv{F}{x_{n}}}$.
    If there is an \eAC  of size $s$ computing $F$, there is an \eAC  of size $O(s)$ computing $F'$.
\end{theorem}

\begin{proof}
    Let $\set{g_1, \dotsc, g_{s}}$ be a minimum size \eAC that computes $F$.
    Let $F^{(i)}$ be the function computed by $g_{i}, \dotsc, g_{s}$.
    Note that $F^{(i)}$ has $s - i + 1$ operations and has input variables $\set{x_{1}, \dotsc, x_{n}, g_{1}, \dotsc, g_{i - 1}}$.
    We prove the desired claim with induction.
    For notational simplicity, we define for $1 \leq j \leq n + s$, $z_{j} = x_{j}$ if $j \leq n$ and $z_{j} = g_{j - n}$ if $j > n$.
    We claim that partial derivatives of $F^{(i)}$ can be computed with an \eAC of size $5(s - i + 1) + 2$. 
    Then, the partial derivatives of $F = F^{(1)}$ can be computed with a circuit of size $O(s)$.

    We begin with the base case.
    Note $F^{(s+1)}(x_{1}, \dotsc, x_{n}, g_{1}, \dotsc, g_{s}) = g_{s}$ has partial derivatives $\pderiv{F^{(s+1)}}{g_{s}} = 1$ and $\pderiv{F^{(s+1)}}{z} = 0$ for any $z \in \set{x_{i}} \cup \set{g_{1}, \dotsc, g_{s - 1}}$. These values are simply constants, so we can compute them using an \eAC of size $1$.
   While we need gates to store the two distinct values, we do not need to store all derivatives separately. In subsequent computations, we simply remember which value to use at the appropriate gate. This leads to $2$ additional gates in the computed circuit.

    We now proceed with the inductive case.
    Assume that $\pderiv{F^{(i + 1)}}{z_{j}}$ is computed for all $j \le n + i$.
    Then, $F^{(i)}$ is the \eAC $F^{(i + 1)}$ with input $g_{i}$ substituted with some operation defined in \Cref{def:extended-arithmetic-circuit}.
    For any variable $z_{j}$ for $j \le n + i - 1$, we can compute $\pderiv{F^{(i)}}{z_{j}}$ by the chain rule,
    \begin{equation*}
        \pderiv{F^{(i)}}{z_{j}} = \pderiv{F^{(i + 1)}}{z_{j}} + \pderiv{F^{(i + 1)}}{g_{i}} \pderiv{g_{i}}{z_{j}} \text{.}
    \end{equation*}
    The case where $g_i$ corresponds to an operation from $\{+, -, \times, /\}$ is already handled in the Morgenstern's proof for standard Baur-Strassen Theorem \cite{DBLP:journals/sigact/Morgenstern85}. Hence, we focus on the case where $g_i$ corresponds to an $\exp$ or $\ln$ gate.

    Suppose $g_{i}$ involves input $z_k$ for some $k \le n + i - 1$, and if $j \ne k$, then  $\pderiv{F^{(i)}}{z_{j}} = \pderiv{F^{(i + 1)}}{z_{j}}$. Hence, it  suffices to compute $\pderiv{F^{(i)}}{z_k}$. 
    
    If $g_{i} = \exp(z_k)$ for some $k \le n + i - 1$, we have
    \begin{equation*}
        \pderiv{F^{(i)}}{z_k} = \pderiv{F^{(i + 1)}}{z_k} + \pderiv{F^{(i + 1)}}{g_{i}} \exp(z_k)
    \end{equation*}
    where we have $\pderiv{F^{(i + 1)}}{z_k}, \pderiv{F^{(i + 1)}}{g_{i}}$ computed by induction and $\exp(z_k)$ can be computed with one operation.
    Summing requires an additional operation.

    Similarly, if $g_{i} = \ln(z_k)$, we have
    \begin{equation*}
        \pderiv{F^{(i)}}{z_k} = \pderiv{F^{(i + 1)}}{z_k} + \pderiv{F^{(i + 1)}}{g_{i}} \frac{1}{z_k}
    \end{equation*}
    where the above term can be computed by $3$ operations.
    This proves the inductive claim.
\end{proof}

Our final ingredient is a lower bound on the size of any arithmetic circuit that computes independent matrix products. 

\begin{lemma}
    \label{lem:independent-product-lb}
    For any integer $\ell \ge 1$, any arithmetic circuit computing $\ell$ independent $n \times n \times n$ matrix products has size $\ell \cdot n^{\omega - o(1)}$. 
\end{lemma}

Its proof follows from standard application of the asymptotic sum inequality \cite{Schonhage81}. We include a proof of this lemma and necessary definitions for completeness. 

A tensor $T$ over three sets of variables $X = \{x_1, \ldots, x_{|X|}\}, Y = \{y_1, \ldots, y_{|Y|}\}, Z = \{z_1, \ldots, z_{|Z|}\}$ is defined as 
\[
T = \sum_{i=1}^{|X|} \sum_{j = 1}^{|Y|} \sum_{k = 1}^{|Z|} a_{ijk} x_i y_j z_k
\]
for some coefficients $a_{ijk} \in \R$. 
An algebraic circuit computing the tensor is an algebraic circuit that, given inputs $X, Y$, outputs $\sum_{i=1}^{|X|} \sum_{j = 1}^{|Y|}  a_{ijk} x_i y_j$ for every $k \in [|Z|]$. 

The rank $R(T)$ of a tensor $T$ is defined as the smallest integer $r$ such that $T$ can be represented as 
\[
\sum_{\ell=1}^r \left(\sum_{i=1}^{|X|} \alpha_{\ell, i} x_i\right) \left(\sum_{j=1}^{|Y|} \beta_{\ell, j} y_j \right) \left(\sum_{k=1}^{|Z|} \gamma_{\ell, k} z_k\right)
\]
for some values $\alpha_{\ell, i}, \beta_{\ell, j}, \gamma_{\ell, k}$. 
The asymptotic rank $\tilde{R}(T)$, which is upper bounded by $R(T)$, is defined as $\lim_{m \rightarrow \infty} (R(T^{\otimes m}))^{1/m}$, where $T^{\otimes m}$ denotes the $m$-th tensor power of $T$. 

The matrix multiplication tensor for $n \times n \times n$ matrix products is
\[
\sum_{i=1}^{n} \sum_{j=1}^{n} \sum_{k = 1}^{n} x_{ij} y_{jk} z_{ik} \text{,}
\]
and we denote it by $\langle n, n, n \rangle$. 
We also denote the disjoint union of $\ell$ copies of $n \times n \times n$ as $\ell \odot \langle n, n, n \rangle$. 

Now we are ready to state the asymptotic sum inequality. 

\begin{theorem}[Asymptotic Sum Inequality \cite{Schonhage81}]
For integers $\ell, n \ge 1$, 
\[
\ell \cdot n^{\omega} \le \tilde{R}\left(\ell \odot \langle n, n, n \rangle\right). 
\]
\end{theorem}

We now prove \Cref{lem:independent-product-lb}.

\begin{proof}[Proof of \cref{lem:independent-product-lb}]

Suppose there is an arithmetic circuit of size $O(\ell \cdot n^{\omega - \eps})$ for $\ell \odot \langle n, n, n \rangle$ for some $\eps > 0$. By definition, this circuit computes $\ell$ independent instances of $n \times n \times n$ matrix products. 
It is known that the rank of a tensor is at most twice the size of a circuit computing it~\cite{strassen1973vermeidung} (see also, e.g, \cite[Theorem 4.7]{DBLP:journals/toc/Blaser13}), so $R(\ell \odot \langle n, n, n \rangle) \le O(\ell \cdot n^{\omega - \eps})$, which further implies $\tilde{R}(\ell \odot \langle n, n, n \rangle) \le O(\ell \cdot n^{\omega - \eps})$. Hence, by the asymptotic sum inequality \cite{Schonhage81}, 
    \begin{align*}
        \omega & \le \frac{\log\left(\tilde{R}(\ell \odot \langle n, n, n \rangle) / \ell\right)}{\log n}\\
        & \le \frac{\log\left(O(n^{\omega - \eps})\right)}{\log n}\\
        & \le \frac{\delta + (\omega - \eps) \log n}{\log n}\\
        & = \omega - \eps + \delta / \log n,
    \end{align*}
    where $\delta$ is a constant depending on the constant hidden in the big-$O$ notation. As we can take $n \rightarrow \infty$, this is a contradiction. Hence, $\ell \odot \langle n, n, n \rangle$ does not have an arithmetic circuit of size $O(\ell \cdot n^{\omega - \eps})$ for any $\eps > 0$, i.e., it has size at least $\ell \cdot n^{\omega - o(1)}$. 
\end{proof}

\section{Lower Bounds for Large Embedding Dimension}

In this section, we prove lower bounds against transformers with large embedding dimension in the extended arithmetic circuit model. 
Recall \cref{thm:linear-embedding}:

\begin{restatable}[Formal \Cref{thm:linear-embedding}]{theorem}{ExactRectangularLB}
    \label{thm:linear-embedding-formal}
    There exists a transformer $T$ with input length $N + 1$, embedding dimension $m = 2N + 3$,  $L$ layers, $H$ attention heads per layer, input dimension and output dimension $2N + 3$ satisfying the following:
    \begin{enumerate}
        \item $T$ is computable by an \eAC, each embedding is computable by an arithmetic circuit of size $O(m)$, and all MLP maps are the identity map.
        \item Any \eAC computing $T$ has size at least $L H N^{\omega - o(1)} - \bigO{LHN^2}$.
    \end{enumerate}
\end{restatable}

We begin with defining a transformer that computes row sums of matrix products.
First, we give a simple attention head that computes the (normalized) row sums of a single matrix product after taking entry-wise exponential.
Our attention head also outputs the normalization factors, which we will later use to recover the entries of the matrix product.

\begin{lemma}
    \label{lemma:square-mm}
    Let $m \geq N$.
    There is an attention unit $\Att_{\mm}$ with input length $N+1$, embedding dimension $m' = \max(m + 1, 2N+3)$, input dimension $2N + 3$ such that for any $N \times m$ matrices $A, B$ and vector $D \in \R^{N}$, there exist embedding weights such that given input $X = \begin{pmatrix}
        Z_{1} & I_{N} & 1_{N \times 1} & 0_{N \times 1} \\
        Z_{2} & 0_{1 \times N} & 0 & 1
    \end{pmatrix} \in \R^{(N + 1) \times (2N + 3)}$
    where $Z_{1} \in \R^{N \times (N + 1)}, Z_{2} \in \R^{1 \times (N + 1)}$ are arbitrary matrices, the output $Y \in \R^{(N + 1) \times m'}$ satisfies
    \begin{equation*}
        Y[i, j] = \begin{cases}
            \frac{N}{N + 1} & \textsf{if $i = N + 1$ and $j = 1$} \\
            \frac{D[j - 1]}{N + 1} & \textsf{if $i = N + 1$ and $1 < j \leq N + 1$}  \\
            \frac{S_{i}}{1 + S_{i}} & \textsf{if $i \in [N]$ and $j = 1$} \\
            \frac{D[j - 1]}{1 + S_{i}} & \textsf{if $i \in [N]$ and $1 < j \leq N + 1$} \\
            0 & \otherwise
        \end{cases}
    \end{equation*}
    where $S_{i} = \sum_{j = 1}^{N} \exp(AB^{\T})_{ij}$ for $i \in [N]$.
    Furthermore, all embedding maps can be computed by arithmetic circuits of size $O(Nm)$.
\end{lemma}

\begin{proof}
    Consider the embeddings
    \begin{equation*}
        Q(X) = \begin{pmatrix}
            A & 0_{N \times 1} \\
            0_{1 \times m} & 0
        \end{pmatrix},\quad 
        K(X) = \begin{pmatrix}
            B & 0_{N \times 1} \\
            0_{1 \times m} & 0
        \end{pmatrix} \in \R^{(N + 1) \times (m + 1)}
    \end{equation*}
    so that
    \begin{equation*}
        \exp\left(Q(X)K(X)^{\T}\right) = \begin{pmatrix}
            \exp(AB^{\T}) & 1_{N \times 1} \\
            1_{1 \times N} & 1
        \end{pmatrix}
    \end{equation*}
    Observe that the embedding maps can be expressed as $Q(X) = X W_{Q}, K(X) = X W_{K}$ for matrices $W_{Q}, W_{K} \in \R^{(2N+3) \times (m + 1)}$:
    \begin{equation*}
        W_{Q} = \begin{pmatrix}
            0_{(N+1) \times m} & 0_{(N+1) \times 1} \\
            A & 0_{N \times 1} \\
            0_{2 \times m} & 0_{2 \times 1}
        \end{pmatrix},\quad W_{K} = \begin{pmatrix}
            0_{(N+1) \times m} & 0_{(N+1) \times 1} \\
            B & 0_{N \times 1} \\
            0_{2 \times m} & 0_{2 \times 1}
        \end{pmatrix}
    \end{equation*}
    If $m' > m + 1$, we can append columns of zeros to $W_{Q}, W_{K}$ and therefore $Q(X), Q(K)$ without changing the value of $Q(X) K(X)^{\T}$.
    Then, consider the embedding
    \begin{equation*}
        V(X) = \begin{pmatrix}
            1_{N \times 1} & 0_{N \times 1} & 0_{N \times 1} & \dots & 0_{N \times 1} \\
            0 & D_{1} & D_{2} & \dots & D_{N}
        \end{pmatrix} \in \R^{(N + 1) \times (N + 1)} \text{.}
    \end{equation*}
    Observe that $V$ can be expressed as $V(X) = X W_{V}$ where
    \begin{equation*}
        W_{V} = \begin{pmatrix}
            0_{(2N+1) \times 1} & 0_{(2N+1) \times 1} & 0_{(2N+1) \times (N -2 )}  & 0_{(2N+1) \times 1} \\
            1 & 0 & 0_{1 \times (N - 2)} & 0 \\
            0 & D_1 & \dots & D_{N} 
        \end{pmatrix}
    \end{equation*}
    As before, if $m' > N + 1$, we can append columns of zeros to $W_{V}$ and therefore $V(X)$, which only appends columns of zeros to the output $Y$.
    Note that $Q(X), K(X), V(X)$ can be computed from $X$ with \eAC (in fact \AC) circuits of size $O(Nm)$ with the relevant matrices $A, B$ and vector $D$ encoded into the weights.
    Then, let $Y \in \R^{(N + 1) \times m'}$ denote the output of attention and $S_{i} := \sum_{j} \exp(AB^{T})_{ij}$.
    Observe that the outputs are as claimed.
\end{proof}

Next, we show how to compute the sum of many matrix products.

\begin{lemma}
    \label{lem:lh-size-transformer-row-sums}
    Let $L, H \geq 1$ and $m \geq N$.
    There is a transformer $T$ with input length $N+1$, embedding dimension $m'=\max(m+1, 2N+3)$, $L$ layers, $H$ heads in each layer, input dimension $2N+3$, output dimension $m'$ such that for any collection of matrices $A_{k}, B_{k} \in \R^{N \times m}$ for $k \in [LH]$ and vectors $D_{k} \in \R^{N}$, there exist embedding weights such that given input $X = \begin{pmatrix}
        0_{N \times (N + 1)} & I_{N} & 1_{N \times 1} & 0_{N \times 1} \\
        0_{1 \times (N + 1)} & 0_{1 \times N} & 0 & 1
    \end{pmatrix}$, the output
    $Y \in \R^{(N + 1) \times m'}$ has entries
    \begin{equation*}
    Y[i, j] = \begin{cases}
            \frac{LHN}{N + 1} & \textsf{if $i = N + 1$ and $j = 1$} \\
            \sum_{k = 1}^{LH} \frac{D_{k}[j - 1]}{N + 1} & \textsf{if $i = N + 1$ and $1 < j \leq N + 1$} \\
            \sum_{k = 1}^{LH} \frac{S_{ki}}{1 + S_{ki}} & \textsf{if $i \in [N]$ and $j = 1$} \\
            \sum_{k = 1}^{LH} \frac{D_{k}[j - 1]}{1 + S_{ki}} & \textsf{if $i \in [N]$ and $1 < j \leq N + 1$} \\
            0 & \otherwise
        \end{cases} 
    \end{equation*}
    where $S_{ki} = \sum_{j} \exp(A_{k} B_{k}^{\T})_{ij}$.
    Furthermore, all embedding maps can be computed by an arithmetic circuit of size $O(LHNm)$ and the MLPs can be taken to be the identity map.
\end{lemma}

The final stipulation on the size of circuits computing the embedding and MLP maps ensures that these computations do not constitute a computational bottleneck for computing the transformer.

\begin{proof}
    First, we remark that in the construction of \Cref{lemma:square-mm}, we may assume the input dimension is in fact $m'$ by appending columns of zeros to the input $X$ and rows of zeros to the weight matrices $W_{Q}, W_{K}, W_{V}$ without changing the computed embeddings $Q(X), K(X), V(X)$.
    
    We construct a transformer that computes $S_0, S_1, \ldots, S_L \in \R^{(N + 1) \times (N + 1)}$ where for $\ell > 0$, $S_{\ell}$ is the matrix satisfying
    \begin{equation*}
        S_{\ell}[i, j] = \begin{cases}
            \frac{\ell HN}{N + 1} & \textsf{if $i = N + 1$ and $j = 1$} \\
            \sum_{k = 1}^{\ell H} \frac{D_{k}[j - 1]}{N + 1} & \textsf{if $i = N + 1$ and $1 < j \leq N + 1$} \\
            \sum_{k = 1}^{\ell H} \frac{S_{ki}}{1 + S_{ki}} & \textsf{if $i \in [N]$ and $j = 1$} \\
            \sum_{k = 1}^{\ell H} \frac{D_{k}[j - 1]}{1 + S_{ki}} & \textsf{if $i \in [N]$ and $1 < j \leq N + 1$} \\
            0 & \otherwise
        \end{cases}
    \end{equation*}
    where $S_{ki} = \sum_{j = 1}^{N} \exp(A_{k} B_{k}^{\T})_{ij}$ and $X$ denotes the input to the transformer.

    We proceed inductively.
    In the base case $S_{0} = X$, the input to the transformer.
    Suppose $S_{\ell - 1}$ is computed.
    For any $h \in [H]$, \Cref{lemma:square-mm} implies that we obtain an attention head $\Att^{(h, \ell)} \in \attn{m+1,2m+1}{N+1}$ that outputs a matrix $Y^{(h, \ell)}$ such that
    \begin{equation*}
        S_{\ell} = S_{\ell - 1} + \sum_{h} Y^{(h, \ell)} \text{.}
    \end{equation*}
    This follows as $Y^{(h, \ell)}$ is only non-zero in the first $N + 1$ columns, and we show in \Cref{lemma:square-mm} that the input to the attention head can be arbitrary in the first $N + 1$ columns.
    
    Summing over $H$ heads and the output of the previous layer, we obtain that the output of the $\ell$-th layer is $S_{\ell}$ as desired. 
    Finally, we observe that $T(X)$ is exactly $S_{L}$, so we define the output MLP $\Psi$ as the identity map, i.e. each row $v$ of $S_{L}$ is mapped to $\psi(v) = v$.
    Note that every embedding can be computed by an arithmetic circuit of size $O(Nm)$.
\end{proof}

We are now ready to prove  \Cref{thm:linear-embedding-formal}.

\begin{proof}[Proof of \Cref{thm:linear-embedding-formal}]
Suppose there is an \eAC of size $s$ that computes the the transformer $T$ with context length $N + 1$, embedding dimension $2N + 3$, $L$ layers, $H$ heads per layer, input and output dimension $2N + 3$ specified in \Cref{lem:lh-size-transformer-row-sums} (note that this transformer satisfies the first desired property).
We construct an \eAC with $LHN^2$ additional inputs, denoted as $C_{kij}$ and $D_{ki}$ where $k \in [LH]$ and $i, j \in [N]$.
We will now describe an \eAC consisting of three parts: a preprocessing stage to construct the input to the transformer, the circuit that simulates the transformer, and a postprocessing stage that sums up the outputs of the transformer. 

We begin with the preprocessing stage.
For $k \in [LH]$, define $C_{k}$ to be the $N \times N$ matrix with entries $C_{k}[i, j] = C_{kij}$ and $D_{k}$ to be the dimension $N$ vector with entries $D_{k}[i] = D_{ki}$.
Then, given $A_{k}, B_{k} \in \R^{N \times N}$ we construct matrices $\tilde{A}_{k}, \tilde{B}_{k} \in \R^{N \times 2N}$ as follows:
\begin{align*}
    \tilde{A}_{k} = \begin{pmatrix}
        A_{k} & C_{k}
    \end{pmatrix},\quad \tilde{B}_{k} = \begin{pmatrix}
        B_{k} & I
    \end{pmatrix}
\end{align*}
Observe that there is an extended arithmetic circuit of size $\bigO{LHN^2}$ that computes $\tilde{A}_{k}, \tilde{B}_{k}$ for all $k \in [LH]$.

Let $T$ be the transformer with input length $N + 1$, embedding dimension $2N + 3$, $L$ layers, $H$ heads per layer, input dimension $2N + 3$, and output dimension $2N + 3$ given by 
\Cref{lem:lh-size-transformer-row-sums} whose output $Y$ is a matrix whose entries are 
\begin{equation*}
    Y_{i,1} = \sum_{k = 1}^{LH} \frac{S_{ki}}{1 + S_{ki}}
\end{equation*}
for $i \in [N]$ and 
\begin{equation*}
    Y_{i,j} = \sum_{k} \frac{D_{k}[j - 1]}{1 + S_{ki}}
\end{equation*}
for $i \in [N]$ and $1 < j \leq N + 1$ where $S_{ki} = \sum_{j} \exp\left(\tilde{A}_{k} \tilde{B}_{k}^{\T}\right)_{ij}$. 
Recall that we denote the size of the \eAC that computes $T$ to be $s$. 
With $N$ additional gates, we can sum up the first column and obtain an \eAC of size $\bigO{s + LHN^2}$ that computes 
\begin{equation*}
    f\left(\set{A_{k}, B_{k}, C_{k}, D_{k}}_{k = 1}^{LH}\right) = \sum_{i = 1}^{N} Y_{i, 1} = \sum_{i = 1}^{N} \sum_{(h, \ell)} Y^{(h, \ell)}_{i, 1} = \sum_{i = 1}^{N} \sum_{k = 1}^{LH} \frac{S_{ki}}{1 + S_{ki}} \text{.}
\end{equation*}
Now, consider the partial derivatives of $f$.
In particular for any $k \in [LH], i_0, j_0 \in [N]$ with $k$ identified with $(h, \ell)$, the partial derivative of $f$ with respect to $C_{k i_0 j_0}$ is
\begin{align*}
    \pderiv{f}{C_{ki_0j_0}} &= \pderiv{}{C_{ki_0j_0 }} \sum_{i = 1}^{N} \sum_{(h, \ell)}^{LH} Y^{(h, \ell)}_{i_, 1} \\
    &= \pderiv{}{C_{ki_0j_0}} Y^{(h, \ell)}_{i_0, 1} \\
    &= \pderiv{}{C_{ki_0j_0}} \frac{S_{ki_0}}{1 + S_{ki_0}} \\
    &= \frac{\pderiv{S_{ki_0}}{C_{ki_0j_0}}}{(1 + S_{ki_0})^2} \\
    &= \frac{\exp((A_{k} B_{k}^{\T})_{i_0j_0} + C_{ki_0j_0})}{(1 + S_{ki_0})^2} 
\end{align*}
where in the second and fifth equality, we only include summands that involve $C_{k i_0 j_0}$, and in the fifth equality, we use
\begin{equation*}
    (\tilde{A}_{k} \tilde{B}_{k}^{\T})_{i0j0} = (A_{k}B_{k}^{\T})_{i0j0} + C_{ki0j0} \text{.}
\end{equation*}

Then, if we set the value of all auxiliary $C_{k i j}$ to $0$, we obtain
\begin{align*}
    \pderiv{f}{C_{ki_0j_0}} &= \frac{\exp((A_{k} B_{k}^{\T})_{i_0j_0})}{(1 + S_{ki_0})^2} \text{.}
\end{align*}
Similarly, by summing up $Y_{i, i + 1}$, we obtain an \eAC that computes the function
\begin{equation*}
    g\left(\set{A_{k}, B_{k}, C_{k}, D_{k}}_{k = 1}^{LH}\right) = \sum_{i = 1}^{N} \sum_{k = 1}^{LH} \frac{D_{k}[i]}{1 + S_{ki}} \text{.}
\end{equation*}
Taking the partial derivative with respect to $D_{k}[i_0]$,
\begin{equation*}
    \pderiv{g}{D_{k}[i_0]} = \pderiv{}{D_{k}[i_0]} Y^{(h, \ell)}[i_0, i_0 + 1]  =  \frac{1}{1 + S_{ki_0}} \text{.}
\end{equation*}
Applying the Extended Baur-Strassen Theorem (\Cref{thm:baur-strassen-extend}) to both $f$ and $g$, there is an \eAC of size $O(s + LHN^2)$ that computes all partial derivatives of $f$ and $g$.
Then, we can easily compute
\begin{equation*}
    \left(A_{k}B_{k}^{\T}\right)_{i_0j_0} = \ln \left(\pderiv{f}{C_{ki_0j_0}} \right)- 2 \ln \left( \pderiv{g}{D_{k}[i_0]} \right)
\end{equation*}
for all $k \in [LH]$ and $i_0, j_0 \in [N]$ with $\bigO{LHN^2}$ additional gates.

So far, we have obtained an \eAC of size $O(s + LHN^2)$ that computes $A_k B_k^{\T}$ for $k \in [LH]$. 
By \cref{thm:ac-simulates-eac}, there also exists an \AC of size $O(s + LHN^2)$  that computes  $A_k B_k^{\T}$ for $k \in [LH]$, i.e., it computes $LH$ independent matrix products of size $N \times N \times N$.  
Then, by \cref{lem:independent-product-lb}, any \AC that computes these independent matrix products has size $LHN^{\omega - o(1)}$. 
Therefore, we must have $s \ge LHN^{\omega - o(1)} - O(LHN^2)$ as desired.

\end{proof}

\section{Discussion and Open Questions}
\label{sec:discussion}

We provide lower bounds against the efficient computation of transformers.
In particular, we show that in the small embedding regime $m = N^{o(1)}$, an $L$ layer transformer with $H$ heads per layer requires $LHN^{2 - o(1)}$ time under the $\kOVf{3}$ Hypothesis, thus establishing a conditional lower bound stating that the standard $LHN^{2 + o(1)}$ algorithm is optimal up to sub-polynomial factors, improving significantly upon the previous lower bounds based only on the hardness of attention \cite{keles2023computational, alman2023fastattentionboundedentries, gupta2026subquadratic}.
In the large embedding regime, we establish that even when the MLP and embedding layers can be computed with small circuits, a circuit computing a transformer can compute many independent matrix products and must be large.
In particular, even when the MLP and embedding functions can be computed by linear sized circuits, computing an $L$ layer transformer with $H$ heads per layer and embedding dimension $N$, a circuit of size $LHN^{\omega\pm o(1)}$ is both necessary and sufficient.

Finally, we conclude with some open questions. 
First, an interesting direction is to extend \Cref{thm:linear-embedding-formal} to any algorithms in the Word-RAM model (i.e. generalizing the lower bound beyond the extended arithmetic circuit model).

Second, while \Cref{thm:linear-embedding-formal} rules out an arithmetic circuit that improves upon the naive algorithm for embedding dimension $m = N$, it is an interesting open question to rule out any arithmetic circuit improving upon the standard algorithm (when sped up with fast matrix multiplication) for any polynomial $m$.
We remark that our techniques imply a lower bound stronger than that of \Cref{thm:small-embedding} for most polynomial $m$.
In particular, we can construct a transformer with embedding dimension $m$ that computes $LH$ independent products of $m \times N$ and $N \times m$ matrices.
Thus, any circuit that computes transformers with embedding dimension $m = N^{a}$ for $a > 0$ requires $LHN^{\omega(1, a, a) - o(1)}$ size.
Here, $\omega(a, b, c)$ denotes the rectangular matrix multiplication exponent, i.e. the smallest constant such that for any $\eps > 0$, there is an arithmetic circuit of size $\bigO{n^{\omega(a, b, c) + \eps}}$ that computes an $n^{a} \times n^{b} \times n^{c}$ matrix product.
The question then is to resolve the gap between this lower bound and the $LHN^{\omega(1, a, 1)+o(1)}$ upper bound.
We remark that even for a single attention head, there is no known reduction from $N \times m \times N$ rectangular matrix multiplication to attention computation with embedding dimension $m$.

Finally, our work focuses on worst-case lower bounds.
Are there input distributions under which we can develop faster algorithms?
A similar question can be asked for structural assumptions: what structural assumptions (beyond those known for attention computation \cite{alman2023fastattentionboundedentries, gupta2026subquadratic}) admit efficient algorithms?

\section*{Acknowledgments}

We would like to thank anonymous reviewers for many helpful comments and suggestions.

\bibliographystyle{alpha}
\bibliography{references}

\newpage
\appendix

\section{Omitted Proofs}
\label{app:omitted-proofs}

We present omitted proofs in this section.

\SoftmaxSimHardmax*

\begin{proof}[Proof of \cref{lem:softmax-sim-hardmax}]
    Fix $c \gg \log N$ for some sufficiently large constant.
    Define embedding maps $Q'(v) = c \cdot Q(v)$ to scale the $Q$ embedding by $c$ and $K'(v) = K(v), V'(v) = V(v)$ to be exactly as in $f$.
    Note that these embedding maps can be expressed as weight matrices that are either the $m \times m$ identity $I_{m}$ or a scalar of $I_{m}$.
    Fix an index $i$ and assume without loss of generality $1 \in I_{\max}(A(X)_{i})$.
    Let $A(X)_{i}$ denote the $i$-th row of $A(X)$.
    Then,
    \begin{equation*}
        \sum_{i' \not\in I_{\max}(A(X)_{i})} \exp(cA(X)_{ii'}) \leq \frac{N}{\exp(c)} \exp(cA(X)_{i1}) = \frac{1}{\poly(N)} \exp(cA(X)_{i1}) \text{.}
    \end{equation*}
    Then, for $i' \in I_{\max}(A(X)_{i})$ note that $\frac{1}{|I_{\max}(A_{i})|} \geq \sm(cA(X))_{ii'}$ so that
    \begin{align*}
        \left| \sm(cA(X))_{ii'} - \hm(cA(X))_{ii'} \right| &= \frac{1}{|I_{\max}(A(X)_{i})|} - \frac{\exp(cA(X))_{i1}}{\sum_{i'} \exp(cA(X))_{ii'}} \text{.}
    \end{align*}
    Now, we can upper bound
    \begin{equation*}
        \sum_{i'} \exp(cA(X))_{ii'} \leq |I_{\max}(A(X)_{i}| \exp(cA(X))_{i1} + \frac{1}{\poly(N)}
    \end{equation*}
    so that 
    \begin{align*}
        \left| \sm(cA(X))_{ii'} - \hm(cA(X))_{ii'} \right| &\leq \frac{1}{|I_{\max}(A(X)_{i})|} \left( 1 - \frac{1}{1 + \frac{1}{\poly(N)}} \right) \leq \frac{1}{\poly(N)} \text{.}
    \end{align*}
    For $i' \not\in I_{\max}(A(X)_{i})$, note that $\hm(cA(X))_{ii'} = 0$ so that
    \begin{align*}
        \left| \sm(cA(X))_{ii'} - \hm(cA(X))_{ii'} \right| &\leq \frac{\exp(cA(X))_{ii'}}{|I_{\max}|\exp(cA(X))_{i1}} \leq \frac{1}{\poly(N)} \text{.}
    \end{align*}
    Then, for a fixed column of $V(X)$, denoted $V(X)^{j}$, we use Holder's inequality to obtain
    \begin{align*}
        \left| f'(X)_{ij} - \hm(cA(X)_{i}) \cdot  V(X)^{j} \right| &= \left| (\sm(cA(X)_{i}) - \hm(cA(X)_{i})) \cdot  V(X)^{j} \right| \\
        &\leq \norm{ \sm(cA(X)_{i}) - \hm(cA(X)_{i})}_{1} \norm{V(X)^{j}}_{\infty} \\
        &\leq \frac{1}{\poly(N)} \text{.}
    \end{align*}
    In particular, every entry of the output is within $\frac{1}{\poly(N)}$ for an arbitrary small inverse polynomial by choosing a larger constant factor for $c$.
\end{proof}

\UnbalancedKOVEquiv*

\begin{proof}[Proof of \Cref{lemma:unbalanced-k-OV}]
    Suppose we are given an instance to (balanced) $\kOV$ of size $n$, with sets $B_{1}, \dots, B_{k}$.
    For $2 \leq j \leq k$, divide $B_{j}$ into $t_{j} := O(n^{1 - s_{j}})$ sets of size $n^{s_{j}}$, denoted $B_{j}^{(1)}, \dots, B_{j}^{(t_{j})}$.
    Now, for every $k$-tuple $(i_1, \dots, i_{k}) \in \prod_{j = 1}^{k} [t_{j}]$, construct an instance of unbalanced $\kOV$ consisting of sets $B_{1}, B_{2}^{(i_2)}, \dots, B_{k}^{(i_{k})}$.
    Note that this creates $\bigO{\prod_{j = 2}^{k} t_{j}}$ instances of unbalanced $\kOV$.
    Suppose there is an algorithm computing unbalanced $\kOV$ in time $\bigO{n^{s_{1} + \dotsc + s_{k} - \eps}}$ for some $\eps > 0$.
    We can then compute all of the above instances in time
    \begin{equation*}
        \bigO{\prod_{j = 2}^{k} t_{j} n^{1 + s_{2} + \dots + s_{k} - \eps}} = \bigO{n^{k - \eps}} \text{.}
    \end{equation*}
    Since, there is an orthogonal $k$-tuple in the original $\kOV$ instance if and only if one of the unbalanced $\kOV$ instances contains an orthogonal $k$-tuple, this violates the $\kOV$ Hypothesis.
\end{proof}

\subsection{Denormalized Attention}

When thinking about attention, it is frequently useful to be able to ignore the normalization required by softmax.
To this end, we define denormalized attention, which replaces $\sm(v)$ with $\exp(v)$ for every row $v$ of $Q(X) K(X)^{\T}$.

\begin{definition}[Denormalized Attention]
    \label{def:denormalized-attn}
    An \emph{denormalized attention head} is a mapping $g_{Q, K, V}: \R^{N \times \din} \to \R^{N \times m}$ defined by
    \[g_{Q, K, V}(X)= \exp(Q(X) K(X)^\T) V(X) \]
    and parameterized by row-wise
    \emph{query}, \emph{key}, and \emph{value embeddings} $Q, K, V \colon \R^{N\times \din} \to \R^{N \times m}$ (e.g., $Q(X) = (Q_1(X_1), \dots, Q_N(X_N))$.
\end{definition}

We show that denormalized attention can be efficiently computed with a small transformer in the standard setting (i.e. with normalization).
Unfortunately, this does not in general imply that a transformer with normalization can simulate a transformer without normalization, as the MLP is applied to the output of each layer (i.e. after summation over heads) whereas the construction below requires an MLP applied immediately to the output of each attention head.

For simplicity, we describe the following transformer 

\begin{lemma}
    \label{lem:denormalized-attn}
    Let $g_{Q, K, V}$ be a denormalized attention head with input length $N$, input dimension $\din$ and embedding dimension $m$.
    There is a transformer $T$ with input length $N + 1$, embedding dimension $m + \din + 1$, $1$ layer, $1$ head, input dimension $\din$ and output dimension $m$ such that 
    \begin{equation*}
        T\left(\begin{pmatrix}
        X\\
        0
    \end{pmatrix}\right) = \begin{pmatrix}
        g_{Q, K, V}(X) \\
        1_{1 \times N} V(X)
    \end{pmatrix}
    \end{equation*}
    for all $X \in \R^{N \times \din}$.
    Furthermore, all embedding and MLP maps can be computed by an arithmetic circuit of size $O(N m)$.
\end{lemma}

\begin{proof}
    Consider the following embeddings:
    \begin{align*}
        Q'(X) &= \begin{pmatrix}
            Q(X) & 0_{N \times (\din + 1)} \\
            0_{1 \times m} & 0_{1 \times (\din + 1)}
        \end{pmatrix},\quad \\
        K'(X) &= \begin{pmatrix}
            K(X) & 0_{N \times (\din + 1)} \\
            0_{1 \times m} & 0_{1 \times (\din + 1)}
        \end{pmatrix},\quad \\
        V'(X) &= \begin{pmatrix}
            0_{N \times \din} & V(X) & 0_{N \times 1} \\
            0_{1 \times \din} & 0_{1 \times m} & 1
        \end{pmatrix}
    \end{align*}
    where $Q'(X), K'(X), V'(X) \in \R^{(N + 1) \times (m + \din + 1)}$.
    Then
    \begin{equation*}
        \exp(Q'(X) K'(X)^{\T})= \begin{pmatrix}
            \exp(Q(X) K(X)^{\T}) & 1_{N \times 1} \\
            1_{1 \times N} & 1
        \end{pmatrix} \text{.}
    \end{equation*}
    In particular, the first $N$ rows have row-sum $S_{i} + 1$ where $S_{i}$ is the row-sum of the $i$-th row in the denormalized attention head.
    Note $S_{i} \geq 0$ so the denominator is always positive.
    Thus, in the output matrix, the first $N$ rows consist of the input $X^{(0)} = \begin{pmatrix}
        X \\
        0
    \end{pmatrix}$.
    Then, for $i \in [N]$, the $(i, \din + m +1)$-th entry is $\frac{1}{1 + S_{i}}$ while for $j \in [m]$ the $(i, \din + j)$-th entry is $\frac{Y_{i, j}}{1 + S_{i}}$ where $Y_{ij}$ is the $(i, j)$-th entry of the denormalized attention head $g$.
    For $i = N + 1$, observe the $(N + 1, \din + m + 1)$-th entry is $\frac{1}{N + 1}$ and for $j \in [m]$, the $(N + 1, j)$-th entry is $\frac{\norm{V(X)[;j]}_{1}}{N + 1}$.
    In particular, we obtain the vector $1_{1 \times N} V(X)$ scaled by $\frac{1}{N + 1}$.
    
    Then, applying the output MLP $\psi = (\psi_{1}, \dotsc, \psi_{m}): \R^{\din + m + 1} \rightarrow \R^{m}$ defined by $\psi_{i}: v \mapsto v_{\din + i} / v_{\din + m + 1}$, we obtain the desired output $Y$.
    Here, note that $v_{\din + m + 1} = \frac{1}{1 + S_{i}} > 0$ since $S_{i} > 0$.
    Note that every embedding map and MLP can be implemented with a size $O(m)$ \eAC.
\end{proof}

\section{MLPs}
\label{app:mlp-examples}

In standard transformer architectures, MLPs are applied after the attention layers. 
In our case, $X = \sum_{h = 1}^{H} \Att_{i}(X')$ is the sum of the outputs of the attention heads of the corresponding layer. 
MLPs are functions that are applied to input $X$. 

\paragraph{GLU MLPs.}
In modern language models, gated linear units (GLUs) are a popular choice of activation function \cite{shazeer2020glu}, for example in PaLM \cite{DBLP:journals/jmlr/ChowdheryNDBMRBCSGSSTMRBTSPRDHPBAI23}, LLaMA \cite{DBLP:journals/corr/abs-2407-21783}, and Gemma \cite{team2024gemma}.
We will show that without loss of generality, transformers with GLU MLPs can simulate transformers with no MLPs, thus obtaining our lower bounds for transformers with GLU MLPs.
We define a gated linear unit below.

\begin{definition}[Gated Linear Unit (GLU) MLP]
    \label{def:glu}
    Given an activation function $\sigma: \R \rightarrow \R$, an GLU with hidden dimension $\dhid$ and activation function $\sigma$ is a map $\psi: \R^{m} \rightarrow \R^{m}$ defined
    \begin{equation*}
        x \mapsto \begin{pmatrix}
            \begin{pmatrix} x^{\T} & 1 \end{pmatrix} W_{0} \odot \sigma\left(\begin{pmatrix} x^{\T} & 1 \end{pmatrix} W_{1}\right) & 1
        \end{pmatrix} W_{2}
    \end{equation*}
    where $W_{0}, W_{1}\in \R^{(m + 1) \times \dhid}$, $W_{2} \in \R^{(\dhid + 1) \times m}$ and $\sigma$ is applied entry-wise.
\end{definition}

We show that transformers with GLU MLPs can simulate transformers with no MLPs whenever the activation function $\sigma$ is not identically $0$.
In particular, our lower bounds for \Cref{thm:small-embedding-formal} and \Cref{thm:linear-embedding-formal} hold also for transformers the GLU MLPs.

\begin{proposition}
    \label{prop:glu-sim-no-mlp}
    Suppose there exists some $c \in \R$ with $\sigma(c) \neq 0$.

    Let $T$ be a transformer with input length $N$, embedding dimension $m$, $L$ layers, $H$ heads, input dimension $\din$, output dimension $\dout$ with no MLPs.

    Then, there is a transformer $T'$ with input length $N$, embedding dimension $m$, $L$ layers, $H$ heads, input dimension $\din$, output dimension $\dout$ and GLU MLPs with $\dhid = m$ and activation function $\sigma$ such that $T'(X) = T(X)$ for all inputs $X$.

    Furthermore, if $\sigma$ can be computed by an \eAC of size $t$, the GLU MLPs can be computed with an \eAC of size $O(tm)$.
\end{proposition}

\begin{proof}
    Let $\sigma(c) = z \neq 0$.
    We construct a GLU MLP that computes the identity map.
    We set $\dhid = m$ and
    \begin{equation*}
        W_{0} = \begin{pmatrix}
            z^{-1} \cdot I_{m} \\
            0_{1 \times m}
        \end{pmatrix}, \quad W_{1} = \begin{pmatrix}
            0_{m \times m} \\
            c \cdot 1_{1 \times m}
        \end{pmatrix}, \quad W_{2} = \begin{pmatrix}
            I_{m} \\
            0_{1 \times m}
        \end{pmatrix} \text{.}
    \end{equation*}
    Then, for any $x$, the GLU MLP maps
    \begin{equation*}
        x \mapsto \frac{x}{z} \odot (\sigma(c) 1_{m}) = x
    \end{equation*}
    so that the GLU MLP $\psi$ is the identity map.
\end{proof}

In the remainder of this section, we show that our lower bounds also hold for traditional (non-gated) MLPs. 
Traditionally, an MLP layer, given some input $X \in \R^{N \times m}$, computes $f(XW_1)W_2$ where $W_1, W_2 \in \R^{m \times \dhid}$ are learned weight matrices and $f$ is an entry-wise activation function.
In general, when we multiply by $W_1, W_2$, we may assume that the input matrix has a column of $1$'s appended to allow for a bias term to be added to the linear function computed by $W_1, W_2$.
When it is clear from context, for matrix $X \in \R^{N \times m}$, we abuse notation and let
\begin{equation*}
    X W := \begin{pmatrix}
        X & 1_{N \times 1}
    \end{pmatrix} W
\end{equation*}
where $W \in \R^{(m + 1) \times \dhid}$.

\begin{definition}[MLP]
    \label{def:typical-mlp}
    Given an activation function $\sigma: \R \rightarrow \R$, an MLP with hidden dimension $\dhid$ and activation function $\sigma$ is a map $\psi: \R^{m} \rightarrow \R^{m}$ defined
    \begin{equation*}
        x \mapsto \begin{pmatrix}
            \sigma\left(\begin{pmatrix} x^{\T} & 1 \end{pmatrix} W_{1}\right) & 1
        \end{pmatrix} W_{2}
    \end{equation*}
    where $W_{1} \in \R^{(m + 1) \times \dhid}$, $W_{2} \in \R^{(\dhid + 1) \times m}$ and $\sigma$ is applied entry-wise.
\end{definition}
In particular, we consider two standard choices of activation function: ReLU and sigmoid.

\paragraph{ReLU. }
For the ReLU activation function $x \mapsto \max(0, x)$, we claim that a  transformer with ReLU MLP layers can simulate a  transformer with no MLPs.
Our lower bound then implies that there are transformers with ReLU MLP layers that cannot be computed with an \eAC of $LHN^{\omega - c}$ for any $c > 0$.
In this case however, it is not clear how to efficiently simulate a ReLU function using an \eAC, so it may be the case that an even larger \eAC is necessary.

\begin{lemma}
    \label{lem:relu-simulates-no-mlp}
    Let $T$ be a transformer with input length $N$, embedding dimension $m$, $L$ layers, $H$ heads, input dimension $\din$, output dimension $\dout$ with no MLPs.

    Then, there is a transformer $T'$ with input length $N$, embedding dimension $m$, $L$ layers, $H$ heads, input dimension $\din$, output dimension $\dout$ and MLPs with $\dhid = 2m$ and ReLU activation functions such that $T'(X) = T(X)$ for all inputs $X$.
\end{lemma}

\begin{proof}
    Given $X \in \R^{N \times m}$, we show how to use an MLP with ReLU activation to compute the identity map.
    Let
    \begin{equation*}
        W_{1} = \begin{pmatrix}
            I_{m} & - I_{m} \\
            0_{1 \times m} & 0_{1 \times m}
        \end{pmatrix},\quad W_{2} = \begin{pmatrix}
            I_{m} \\
            - I_{m} \\
            0_{1 \times m}
        \end{pmatrix} \text{.}
    \end{equation*}
    Now, letting $f(x) = \max(x, 0)$ denote the ReLU function, observe
    \begin{equation*}
        f(XW_{1})W_{2} = \begin{pmatrix}
            f(X) & f(-X)
        \end{pmatrix} W_{2} = f(X) - f(-X) \text{.}
    \end{equation*}
    We conclude the proof by noting that $f(x) - f(-x) = \max(0, x) - \max(0, -x) = x$ for all $x \in \R$.
\end{proof}

\paragraph{Sigmoid. }
For the sigmoid activation function (denoted $\sigma(x) = 1/(1 + e^{-x})$), we require a more intricate modification of \Cref{thm:small-embedding-formal} and \Cref{thm:linear-embedding-formal}.
We will in fact prove the results for a more general class of activation functions.

\begin{definition}
    \label{def:gap-preserving-activation}
    An activation function $\sigma: \R \rightarrow \R$ is gap-preserving if the following hold:
    \begin{enumerate}
        \item $\sigma(0) \neq \sigma(1)$,
        \item There exists $c^*$ such that for all $x \leq c^* - 1$, $\sigma(x) \leq \sigma(c^* - 1) < \sigma(c^*)$.
    \end{enumerate}
\end{definition}

Note that the sigmoid is gap-preserving with $c^* = 0$.
We remark that the specific choices of constants (e.g. $0, 1, -1, -.1$) are not too important, and the proof can be modified for other choices of constants.
We show that our lower bound in the small embedding dimension (\Cref{thm:small-embedding-formal}) holds for MLPs with gap-preserving activation functions.

\begin{proposition}
    \label{prop:small-embedding-sigmoid}
    Let $m = \Theta(\log N)$ and $L, H = \poly(N)$.
    Any algorithm computing transformers with input length $N$, embedding dimension $m$, $L$ layers, $H$ heads in each layer, input dimension $m$, and output dimension $1$ up to entry-wise additive error $\frac{1}{10N}$ requires $LHN^{2 - o(1)}$ time under the $\kOVf{3}$ Hypothesis.
    The lower bound holds even if the transformer has MLPs with hidden dimension $O(m)$ and gap-preserving activation function.
\end{proposition}

We give a high level overview of the modifications.
Recall that in \Cref{thm:small-embedding-formal} we construct a transformer such that each layer's output (1) preserves its input on $m-1$ columns and (2) maintains one column to record whether an orthogonal triple has been found.
On this final column, the value is $2H\ell$ if $a_i$ has not participated in any orthogonal triple so far and at most $2H\ell - 0.5$ otherwise, achieving a $0.5$ additive gap between the yes and no instances.
We thus design an MLP that computes the identity map on $m-1$ columns and maintains a gap on the final column.
To compute the identity, we observe that each of the first $m-1$ columns has at most $2$ distinct values, thus we can design an appropriate linear function to map the outputs of the sigmoid $\sigma$ back to the input.
In the final column, we construct $W_1$ to map $x \mapsto x - 2H$ and $W_2$ to map $\sigma(0) \mapsto 0$ and $\sigma(-1/2) \mapsto y < -0.5$, maintaining a gap of at least $0.5$ in the output.

\begin{proof}
    We describe the appropriate modifications to \Cref{thm:small-embedding-formal}, beginning with the first layer.
    Let $X^{(0)}$ denote the input to the transformer and $Z_{h, \ell}$ denote the output of $h$-th head in the $\ell$-th layer of the hardmax transformer defined in \Cref{thm:small-embedding-formal}.
    Recall
    \begin{equation*}
        X^{(0)} = \begin{pmatrix}
            a_{1}^{\T} & b_{1}^{\T} & 1 \\
            \vdots & \vdots & \vdots \\
            a_{N}^{\T} & b_{N}^{\T} & 1 \\
            0 & 0 & 2
        \end{pmatrix} \in \R^{(N + 1) \times (m - 1)}, \quad Z_{h, \ell} = \begin{pmatrix}
            0 & \dots & 0 & z_{1, h, \ell} \\
            \vdots & \ddots & \vdots & \vdots \\
            0 & \dots & 0 & z_{N, h, \ell} \\
            0 & \dots & 0 & \frac{N + 2}{N + 1} 
        \end{pmatrix}\in \R^{(N + 1) \times m}
    \end{equation*}
    where $z_{i, h, \ell} = 2$ if $a_{i}, c_{h, \ell}$ do not participate in an orthogonal triple, and $z_{i, h, \ell} \leq \frac{3}{2}$ otherwise.
    Then,
    \begin{equation*}
        \sum_{h} Z_{h, \ell} = \begin{pmatrix}
            0 & \dots & 0 & z_{1, \ell} \\
            \vdots & \ddots & \vdots & \vdots \\
            0 & \dots & 0 & z_{N, \ell} \\
            0 & \dots & 0 & \frac{H(N + 2)}{N + 1} 
        \end{pmatrix}
    \end{equation*}
    where $z_{i, \ell} = \sum_{h} z_{i, h, \ell}$.
    
    We construct an MLP with hidden dimension $\dhid = m$ and gap-preserving activation function.
    Define linear functions $f_1, f_2, f_3$ such that
    $f_1$ is an increasing linear function mapping $\sigma(c^*) \mapsto 0$ and $\sigma(c^* - 1) \mapsto -1$ and
    \begin{align*}
        f_1(x) &= b_{1} x + b_{2} \\
        f_2(x) &= \frac{x - \sigma(0)}{\sigma(1) - \sigma(0)}  = c_1 x + c_2 \\
        f_3(x) &= 1 + f_2(x) = c_1 x + (1 + c_2)
    \end{align*}
    for some constants $b_1, b_2, c_1, c_2 \in \R$ and observe
    \begin{align*}
        f_1&: \sigma(c^*) \mapsto 0, \quad \sigma(c^*-1) \mapsto - 1 \\
        f_2&: \sigma(0) \mapsto 0, \quad \sigma(1) \mapsto 1 \\
        f_3&: \sigma(0) \mapsto 1, \quad \sigma(1) \mapsto 2 \text{.}
    \end{align*}
    Note that we can construct $f_1, f_2, f_3$ by the first two properties of \Cref{def:gap-preserving-activation}.
    Set the weight matrices as
    \begin{equation*}
        W_{1} = \begin{pmatrix}
            I_{2d} & 0_{2d \times 1} & 0_{2d \times 1} \\
            0_{1 \times 2d} & 1 & 0 \\
            0_{1 \times 2d} & 0 & 2 \\
            0_{1 \times 2d} & -1 & c^* - 4H
        \end{pmatrix},\quad W_{2} = \begin{pmatrix}
            c_1 I_{2d} & 0_{2d \times 1} & 0_{2d \times 1} \\
            0_{1 \times 2d} & c_1 & 0 \\
            0_{1 \times 2d} & 0 & b_1 \\
            (c_2)_{1 \times 2d} & 1 + c_2 & b_2
        \end{pmatrix} \text{.}
    \end{equation*}
    Note that $W_{1}, W_{2}$ have $m + 1$ rows (instead of $m$) since the last row encodes the bias term of the linear function $W_1, W_2$.
    
    Since the first $2d$ columns of $X^{(0)}$ is $0$ or $1$, we ensure that the first $2d$ columns is preserved by the MLP, since $W_1$ is the identity map and $f_2$ maps $\sigma(0) = \frac{1}{2}$ to $0$ and $\sigma(1)$ to $1$.
    
    Similarly, since the $(2d + 1)$-th column is $1$ or $2$, we ensure that this is preserved by the MLP as well. 
    Here, observe that $W_1$ maps $1$ and $2$ to $0$ and $1$ respectively, while $f_3$ maps $\sigma(0)$ to $1$ and $\sigma(1)$ to $2$.

    Finally, consider an entry in the $(2d+2)$-th column.
    As before, if $a_i$ does not participate in an orthogonal triple with any $\set{c_{h, 1}}$ then $z_{i, 1} = 2H$ and otherwise, $z_{i, 1} \leq 2H - \frac{1}{2}$.
    Then, $\sigma(Y^{(1)}W_1)[i,m] = \sigma(c^*)$ in the former case and $\sigma(x) \leq \sigma(c^* - 1)$ in the latter case.
    In particular, the MLP layer outputs $X^{(1)}[i, m] = (\sigma(Y^{(1)}W_1)W_2)[i, m] = 0$ in the former case, and is at most $-1$ (since $f_1$ is increasing) in the latter case. 

    We can proceed inductively for subsequent layers.
    Note that with the exception of the $m$-th column, the identical proof as above shows that the contents of the input matrix $X^{(0)}$ are preserved.
    It suffices to show that the $m$-th column allows us to distinguish a yes-instance of $\kOVf{3}$ from a no-instance.
    We use the following inductive claim:
    for all $i$, if $a_i$ participates in an orthogonal triple with $c_{h, t}$ for some $h \in [H], t \leq \ell$, then $X^{(\ell)}[i, m] \leq -1$. Otherwise, $X^{(\ell)}[i, m] = 0$.
    Above we have established the base case with $\ell = 1$.
    
    Let $Y^{(\ell)}$ denote the output of the attention layer, we have
    \begin{equation*}
        Y^{(\ell)}[i, m] = X^{(\ell - 1)}[i, m] + \sum_{h} z_{i, h, \ell} \text{.}
    \end{equation*}
    In a no-instance, we have $X^{(\ell - 1)}[i, m] = 0$ (by induction) and $\sum_{h} z_{i, h, \ell} = 2H$.
    Following the proof of the base case, we conclude $X^{(\ell)}[i, m] = 0$.
    
    We now argue that in a yes-instance, whenever $a_{i}$ participate in an orthogonal triple with some $c_{h, t}$ for $t \leq \ell$, $X^{(\ell)}[i, m] \leq -1$.
    By induction, we either have $X^{(\ell - 1)}[i, m] \leq -1$ if $t < \ell$ or $\sum_{h} z_{i, h, \ell} \leq 2H - 0.5$. 
    Since $X^{(\ell - 1)}[i, m] \leq 0$ and $\sum_{h} z_{i, h, \ell} \leq 2H$ in either case, we have $Y^{(\ell)}[i, m] \leq 2H - 0.5$.
    Since $W_1$ applies an increasing function to the $m$-th column, we again have that the input to $\sigma$ is $c^*$ in the no-instance and at most $c^*-1$ in the yes instance.
    We then conclude the proof identically as in the base case.

    Thus, encoding MLPs as above, it suffices to examine the $m$-th column of the output of the transformer to compute $\kOVf{3}$.
    In a no instance, all the output entries should be $0$.
    In a yes instance, at least one entry will be less than $-1$.
\end{proof}

We now move on to \Cref{thm:linear-embedding-formal}, which we prove a separate class of activation functions (which sigmoid also satisfies).

\begin{definition}
    \label{def:differentiable-activation}
    An activation function $\sigma: \R \rightarrow \R$ is \eAC gradient efficient if the following hold:
    \begin{enumerate}
        \item $\sigma(0) \neq \sigma(1)$,
        \item $\sigma^{-1}$ can be implemented with an \eAC of size $O(1)$.
        \item $\sigma'(x) = h(\sigma(x))$ for some $h$ such that both $\sigma, h$ can be implemented with an \eAC of size $O(1)$.
    \end{enumerate}
\end{definition}

Note that the sigmoid function is \eAC gradient efficient as $\sigma'(x) = \sigma(x)(1 - \sigma(x))$ and $\sigma^{-1}$ is computable on the range of the sigmoid function $(0, 1)$.
We remark that in the event that the last layer has no MLP, the second condition is unnecessary.

\begin{proposition}
    \label{prop:large-embedding-sigmoid}
    There exists a transformer $T$ with input length $N + 1$, embedding dimension $m = 2N + 1$,  $L$ layers, $H$ attention heads per layer, input dimension $1$ and output dimension $N + 1$ satisfying the following:
    \begin{enumerate}
        \item $T$ is computable by an \eAC, and each embedding and MLP is computable by an \eAC of size $O(m)$.
        \item Any \eAC computing $T$ has size at least $L H N^{\omega - o(1)} - \bigO{LHN^2}$.
    \end{enumerate}
    The lower bound holds even when $T$ has MLPs with $\dhid = O(m)$ and \eAC gradient efficient activation functions.
\end{proposition}

\begin{proof}
    We maintain the construction of \Cref{thm:linear-embedding-formal}, modifying only the MLPs.
    As in \Cref{prop:small-embedding-sigmoid}, each MLP weight matrix has an additional row to allow for a bias term in the linear function.
    We set
    \begin{equation*}
        W_1 = \begin{pmatrix}
            I_{N + 1} & 0_{(N + 1) \times (m' - (N + 1))} \\
            0_{(m' - (N + 1)) \times (N + 1)} & I_{m' - (N + 1)} \\
            0_{1 \times (N + 1)} & 0_{1 \times (m' - (N + 1))}
        \end{pmatrix}, W_2 = \begin{pmatrix}
            I_{N + 1} & 0_{(N + 1) \times (m' - (N + 1))} \\
            0_{(m' - (N + 1)) \times (N + 1)} & c_1 I_{m' - (N + 1)} \\
            0_{1 \times (N + 1)} & - (c_2)_{1 \times (m' - (N + 1))}
        \end{pmatrix}
    \end{equation*}
    where $c_1, c_2$ are constants as defined in \Cref{prop:small-embedding-sigmoid}.
    Note that to define $c_1, c_2$ we only require the property $\sigma(0) \neq \sigma(1)$.
    Thus, the MLP layer applies $\sigma$ entry-wise to the first $N + 1$ columns and maps $0, 1$ to $0, 1$ respectively in the last $m' - (N + 1)$ columns.
    Since $\sigma$ can be computed by an \eAC of $O(1)$ size, each MLP map can be computed by an \eAC of size $O(m)$, satisfying the first property.
    
    We first argue that the outputs of the attention heads as described in \Cref{lemma:square-mm} remain unchanged.
    Since the guarantees of \Cref{lemma:square-mm} holds as long as the last $m' - (N + 1)$ columns are the same as the input $X^{(0)}$, it suffices to argue that this is maintained after each layer, which is formalized in the following claim.

    \begin{claim}
        \label{clm:input-maintain-sigmoid}
        In the transformer defined in \Cref{thm:linear-embedding-formal} with MLPs specified by $W_1, W_2$ above, the output of each layer $\ell \geq 0$ satisfies $X^{(\ell)}[i, j] = X^{(0)}[i, j]$ for all $i \in [N + 1], j > N + 1$ i.e., the input is preserved in the last $m' - N + 1$ columns.
    \end{claim}

    \begin{proof}
        We proceed inductively.
        The base case holds trivially, so we consider the inductive case.
        From \Cref{lemma:square-mm}, each attention head outputs $0$ beyond the first $N + 1$ columns, so that the output of each attention layer is its exactly its input on these columns.
        Since our input to the transformer, denoted $X^{(0)}$, contains only $0, 1$ in these columns, the MLP computes the identity map on these values, and maintains the input beyond the first $N + 1$ columns, proving the inductive hypothesis.
    \end{proof}

    We now argue that we may recover the matrix products efficiently from an \eAC that computes the described transformer.
    Let $h$ be the function such that $\sigma'(x) = h(\sigma(x))$.
    In the following, assume $j \leq N + 1$, since we only need the entries from the first $N + 1$ columns of the output of the transformer to recover the matrix products.
    As in \Cref{prop:small-embedding-sigmoid}, we define
    \begin{equation*}
        Y^{(\ell)} := X^{(\ell - 1)} + \sum_{h} Y^{(h, \ell)}
    \end{equation*}
    where $Y^{(h, \ell)}$ denotes the output of the $h$-th head in the $\ell$-th layer.
    Recall that above we have shown $Y{(h, \ell)}$ are exactly the outputs of the corresponding attention head in the transformer constructed in \Cref{thm:linear-embedding-formal}.
    
    In particular, if $Y$ denotes the output of the transformer, we have for $j \leq N + 1$,
    \begin{align*}
        Y_{i, j} &= \sigma\left(Y^{(L)}_{i, j}\right) \\
        &= \sigma\left(X^{(L - 1)}_{i, j} + \sum_{h} Y^{(h, L)}_{i, j} \right)
    \end{align*}
    and for all $\ell > 0$ and $j \leq N + 1$,
    \begin{align*}
        X^{(\ell)}_{i, j} &= \sigma\left(Y^{(\ell)}_{i, j} \right) \\
        &= \sigma\left(X^{(\ell - 1)}_{i, j} + \sum_{h} Y^{(h, \ell)}_{i, j} \right) \text{.}
    \end{align*}
    In contrast to when there is no MLP, where $Y_{i, j} = \sum_{h, \ell} Y^{(h, \ell)}_{i, j}$ we now have
    \begin{equation}
        \label{eq:sigmoid-mlp-output}
        Y_{i, j} = \sigma\left(\sum_{h} Y^{(h, L)}_{i, j} + \sigma \left( \sum_{h} Y^{(h, L - 1)}_{i, j} + \sigma \left( \dots \right) \right) \right) \text{.}
    \end{equation}
    Recall that in the proof of \Cref{thm:linear-embedding-formal} we use the outputs of the transformer to compute the following two functions:
    \begin{align*}
        f\left( \set{A_{k}, B_{k}, C_{k}, D_{k}}_{k = 1}^{LH} \right) &= \sum_{i = 1}^{N} Y_{i, 1} \\
        g\left( \set{A_{k}, B_{k}, C_{k}, D_{k}}_{k = 1}^{LH} \right) &= \sum_{i = 1}^{N} Y_{i, i + 1} \text{.}
    \end{align*}
    
    We now describe how to recover $LH$ independent matrix products $A_{k}B_{k}^{\T}$, thus completing the proof.
    This is formalized in the following inductive lemma.

    \begin{lemma}
        \label{lem:sigmoid-recover-derivatives}
        Recall that we index $k \in [LH]$ as $(h, \ell)$.
        For all $0 \leq \ell \leq L$, there is an \eAC of size $O(H(L - \ell)N^2)$ that computes 
        $\set{A_{k} B_{k}^{\T}}$ for $k = (h, p)$ with $p \geq \ell$, 
        $\set{\prod_{p = t}^{L} \sigma'\left(Y^{(p)}_{i_0, 1}\right)}_{t \in [\ell - 1, L], i_0}$, 
        and $\set{\prod_{p = t}^{L} \sigma'\left(Y^{(p)}_{i_0, i_0 + 1}\right)}_{t \in [\ell - 1, L], i_0}$.
    \end{lemma}

    \begin{proof}
        For the base case, fix $\ell = L$.
        Then, for fixed $k = (h_0, L)$ and $i_0, j_0$ we have
        \begin{align*}
            \pderiv{f}{C_{ki_0j_0}} &= \pderiv{}{C_{ki_0j_0}} \sum_{i = 1}^{N} Y_{i, 1} \\
            &= \pderiv{}{C_{ki_0j_0}} Y_{i_0, 1} \\
            &= \sigma' \left( Y^{(L)}_{i_0,1} \right) \pderiv{Y^{(h_0, L)}_{i_0, 1}}{C_{ki_0j_0}}  \text{.}
        \end{align*}
        In the second equality, we keep only summands depending on $C_{ki_0j_0}$ and in the third, we apply the chain rule and observe that only $Y^{(h_0, \ell)}_{i_0, 1}$ in $Y_{i_0, 1}^{(L)}$ depends on $C_{ki_0j_0}$.
        Since $\sigma$ is \eAC gradient efficient, we can compute $\sigma'\left(Y_{i_{0}, 1}^{(L)}\right)$ with $h\left(\sigma\left(Y_{i_{0}, 1}^{(L)}\right)\right) = h(Y_{i_0, 1})$ using the output of the transformer.
        Over all $h \in [H]$ and $i_0 \in [N]$, we use an \eAC of size $O(HN)$.
        In particular, we can compute
        \begin{equation*}
            \pderiv{Y^{(h_0, L)}_{i_0, 1}}{C_{ki_0j_0}} = \frac{\pderiv{f}{C_{ki_0j_0}}}{h(Y_{i_0, 1})} \text{.}
        \end{equation*}

        We make a similar argument for $g$ as follows.
        \begin{equation*}
            \pderiv{g}{D_{ki_0}} = \pderiv{}{D_{ki_0}} Y_{i_0, i_0 + 1} = \sigma'(Y_{i_0, i_0 + 1}^{(L)}) \pderiv{Y^{(h_0, L)}_{i_0, i_0 + 1}}{D_{ki_0}} \text{.}
        \end{equation*}
        An identical proof allows us to recover $\pderiv{Y^{(h_0, L)}_{i_0, i_0 + 1}}{D_{ki_0}}$ using an \eAC of size $O(HN)$.

        Since in \Cref{thm:linear-embedding-formal}, we showed how to compute the products from $\pderiv{Y^{(h, \ell)}_{i_01}}{C_{ki_0j_0}}$ and $\pderiv{Y^{(h, \ell)}_{i_0,i_0+1}}{D_{ki_0}}$, we can use the above derivatives to recover the matrix products $A_{k} B_{k}^{\T}$ with an \eAC of size $O(N^2)$.
        Repeating the above argument for all $H$ gives us a $O(HN^2)$ \eAC that computes all matrix products in the $L$-th layer.
        Given the resulting matrix products, we can compute all $Y^{(h, L)}_{i_0, 1}$ with an \eAC of size $O(HN^2)$ by summing the rows of each matrix product.
        Then, since
        \begin{equation*}
            Y_{i_0, 1} = \sigma\left( \sum_{h} Y_{i_0, 1}^{(h, L)} + \sigma \left( Y_{i_0, 1}^{(L - 1)} \right) \right)
        \end{equation*} 
        we can compute
        \begin{equation*}
            \sigma'\left(Y_{i_0, 1}^{(L - 1)}\right) = h\left(\sigma \left( Y_{i_0, 1}^{(L - 1)} \right) \right) = h \left( \sigma^{_1}(Y_{i_0, 1}) - \sum_{h} Y_{i_0, 1}^{(h, L)} \right)
        \end{equation*}
        for all $i_0$ with an \eAC of size $O(HN)$.
        Note that this is only step where we require $\sigma$ to be invertible.
        In particular, we have described how to compute the following quantities:
        \begin{equation*}
            \set{\sigma'\left(Y^{(L - 1)}_{i_0, 1}\right) \sigma'\left(Y_{i_{0, 1}}^{(L)}\right)}_{i_0}\text{, } \set{A_{k}B_{k}^{\T}}_{k = (h, L)}
        \end{equation*}
        using an \eAC of size $O(H N^2)$, completing the base case.
    
        Now, consider the inductive case $\ell < L$.
        We argue that all additional required values can be computed with an \eAC of size $O(HN^2)$.
        Applying the chain rule, for $k \gets (h_0, \ell), i_0, j_0$, we have
        \begin{align*}
            \pderiv{f}{C_{ki_0j_0}} &= \pderiv{}{C_{ki_0j_0}} Y_{i_0, 1} \\
            &= \left( \prod_{t = \ell}^{L} \sigma' \left( Y^{(t)}_{i_0,1} \right) \right) \pderiv{Y^{(h_0, \ell)}_{i_0, 1}}{C_{ki_0j_0}} 
        \end{align*}
        following similar logic to the base case.
        By the inductive hypothesis, we have $\prod_{t = \ell}^{L} \sigma' \left( Y^{(t)}_{i_0,1} \right)$ from the previous $(\ell + 1)$-th iteration. 
        Thus, with an \eAC of size $O(HN)$ we can compute for all $i_0$:
        \begin{equation*}
            \pderiv{Y^{(h_0, \ell)}_{i_0, 1}}{C_{ki_0j_0}}  = \frac{\pderiv{f}{C_{ki_0j_0}}}{\prod_{t = \ell}^{L} \sigma' \left( Y^{(t)}_{i_0,1} \right)} \text{.}
        \end{equation*}
        From here, we again follow the proof of \Cref{thm:linear-embedding-formal} to recover
        \begin{equation*}
            \set{A_{k} B_{k}}_{k = (h, \ell)}
        \end{equation*}
        and sum up matrix rows to obtain $\set{Y_{i_0, 1}^{(\ell)}}_{i_0}$ with an \eAC of size $O(HN^2)$.
        Then, since
        \begin{equation*}
            Y_{i_0, 1}^{(\ell)} = \sum_{h} Y_{i_0, 1}^{(h, \ell)} + \sigma \left( Y_{i_0, 1}^{(\ell - 1)} \right)
        \end{equation*}
        we can compute
        \begin{equation*}
            \sigma' \left( Y_{i_0, 1}^{(\ell - 1)} \right) = h\left( \sigma \left( Y_{i_0, 1}^{(\ell - 1)} \right) \right) = h \left( Y_{i_0, 1}^{(\ell)} - \sum_{h} Y_{i_0, 1}^{(h, \ell)} \right)
        \end{equation*}
        for all $i_0$ using an \eAC of size $O(HN)$.
        Using our inductive hypothesis, we compute for all $i_0$,
        \begin{equation*}
            \prod_{t = \ell - 1}^{L} \sigma' \left( Y^{(t)}_{i_0,1} \right) = \sigma' \left( Y_{i_0, 1}^{(\ell - 1)} \right) \prod_{t = \ell}^{L} \sigma' \left( Y^{(t)}_{i_0,1} \right)
        \end{equation*}
        with $O(N)$ size \eAC, completing the inductive case.
        A similar argument holds for the partial derivatives of $g$ with respect to $D_{ki_0}$.
    \end{proof}

    The above lemma allows us to conclude the proof, computing all $LH$ matrix products.
\end{proof}

\section{Head Aggregation}
\label{app:head-aggregation}

Consider a transformer with $L$ layers of $H$ heads with embedding dimension $m$.
There are two standard ways to aggregate the outputs of the heads in the $\ell$-th layer, denoted $X^{(\ell)}$, given the input $X^{(\ell - 1)}$: summation and concatenation.
Summation aggregation is described in \Cref{def:tran}.
In concatenation aggregation (see e.g. \cite{DBLP:conf/nips/VaswaniSPUJGKP17}), each attention head maps each row of dimension $m$ into dimension $m/H$. 
The output is then
\begin{equation*}
    X^{(\ell)} = \Psi_{\ell} \left( X^{(\ell - 1)} + \begin{pmatrix}
        \Att_{Q_{1}, K_{1}, V_{1}}\left(X^{(\ell - 1)}\right) & \Att_{Q_{2}, K_{2}, V_{2}}\left(X^{(\ell - 1)}\right) & \dots & \Att_{Q_{H}, K_{H}, V_{H}}\left(X^{(\ell - 1)}\right) 
    \end{pmatrix} \right)
\end{equation*}
where $\Psi$ denotes the MLP applied after concatenation.

We show such a transformer can simulate a transformer with summation aggregation and embedding dimension $m/H$, so that any lower bound against transformers with summation aggregation and embedding dimension $m/H$ also hold against transformations with concatenation aggregation and embedding dimension $m$.

\begin{proposition}
    \label{prop:concat-sims-sum}
    Assume $H \leq m$ divides $m$.
    Suppose $T$ is a transformer with input length $N$, embedding dimension $m/H$, $L$ layers, $H$ heads per layer, and summation aggregation.

    Then, there is a transformer $T'$ with input length $N + 1$, embedding dimension $m + H$, $L$ layers, $H$ heads per layer and concatenation aggregation such that there is a linear function $f: \R^{m} \rightarrow \R^{m}$ such that $f(T'(X)) = T(X)$ for all inputs $X$.
    Furthermore, the following hold:
    \begin{enumerate}
        \item If $T$ is computable by an \eAC, so is $T'$.
        \item If each embedding map of $T$ is computable by an \AC of size $O(m/H)$, then all embedding maps of $T$ in one layer are computable by an \AC of size $O(m)$.
        \item If $T$ has no MLPs, then neither does $T'$.
        \item If $T$ has MLPs or GLU MLPs with $\dhid = O(m/H)$ and activation function $\sigma$, then $T'$ has an MLP with $\dhid = O(m)$ and activation function $\sigma$ (see \Cref{def:typical-mlp} and \Cref{def:glu}).
    \end{enumerate}
\end{proposition}

In particular, \Cref{thm:small-embedding-formal} implies that whenever $m/H=N^{o(1)}$, the standard $LHN^{2 + o(1)}$ time algorithm to compute transformers with concatenation aggregation is optimal, and \Cref{thm:linear-embedding-formal} implies that whenever $m/H=\Theta(N)$, the corresponding $LHN^{\omega + o(1)}$ time algorithm is optimal.
Furthermore, the reduction preserves whatever assumptions we make on the MLPs, either in the main body or \Cref{app:mlp-examples}.

\begin{proof}
    We show that each layer can be simulated. Let $Y^{(\ell)}$ be the output of the $\ell$-th layer of $T$. In addition, let $Y^{(\ell, h)}$ be the output of the $h$-th attention head in the $\ell$-th layer of $T$ before the summation. By definition, $Y^{(\ell)} = \Psi_{\ell} \left(Y^{(\ell - 1)} + \sum_{h} Y^{(\ell, h)}\right)$ for $\ell \ge 1$ (in the case where $T$ does not have MLP, $\Psi_{\ell}$ is simply the identity map). 
    Consider a single layer of $T'$ with input $X^{(\ell - 1)}$, where we view it as a concatenation of $H$ blocks:
    \begin{equation*}
        \begin{pmatrix}
        X^{(\ell - 1)}_{1} & X^{(\ell - 1)}_{2} & \dots & X^{(\ell - 1)}_{H} 
        \end{pmatrix} \text{.}
    \end{equation*}
    We define the input of the transformer $T'$ as $X_{1}^{(0)} = Y^{(0)}$ and $X_{h}^{(\ell - 1)} := 0_{N \times (m/H)}$ for all $h > 1$.
    We will prove two cases separately: when $T$ has no MLPs, and when $T$ has MLPs with $\dhid = O(m)$ and activation function $\sigma$.

    \paragraph{Case 1: $T$ has no MLPs.}
    We proceed with the inductive hypothesis 
    \begin{equation*}
        \sum_{h = 1}^{H} X_{h}^{(\ell - 1)} = Y^{(\ell - 1)} \text{.}
    \end{equation*}
    That is, the sum of all the blocks is the output of the $(\ell - 1)$-th layer of $T$.
    Note that the base case is satisfied by construction.
    
    Consider the $h$-th attention head.
    Fix an attention head of $T$ with embedding maps $Q_{h}, K_{h}, V_{h}: \R^{N \times (m/H)} \rightarrow \R^{N \times (m/H)}$ (recall that $T$ has embedding dimension $m/H$). 
    We design attention heads with similar embedding maps.
    Define the map $g: \R^{m} \rightarrow \R^{m/H}$ where
    \begin{equation*}
        g: v \mapsto v \begin{pmatrix}
            I_{m/H} \\
            \vdots \\
            I_{m/H}
        \end{pmatrix} := v W_{g}
    \end{equation*}
    so that $g(v)_{i} = \sum_{i \equiv i' \mod (m/H)} v_{i}$.
    
    If $T$ has no MLPs, we design embedding maps $Q'_{h}, K'_{h}, V'_{h}: \R^{N \times m} \rightarrow \R^{N \times (m/H)}$ as
    \begin{align*}
        Q'_{h} = Q_{h} \circ g, K'_{h} = K_{h} \circ g, V'_{h} = V_{h} \circ g \text{.}
    \end{align*}
    Note here that if $Q_{h}: X \rightarrow X W_{Q_{h}}$ can be described as a linear function, so can $Q'_{h}: X \rightarrow X W_{g} W_{Q_{h}}$ (and analogously for $K_{h}, V_{h}$).
    Under the inductive hypothesis, we have that
    \begin{equation*}
        Q'_{h}\left( \begin{pmatrix}
        X^{(\ell - 1)}_{1} & X^{(\ell - 1)}_{2} & \dots & X^{(\ell - 1)}_{H} 
        \end{pmatrix} \right) = Q_{h} \left( \sum_{h}  X^{(\ell - 1)}_{h} \right) = Q_{h} (Y^{(\ell - 1)})
    \end{equation*}
    is exactly the embedding computed by the $h$-th head of the $\ell$-th layer of $T$.
    A similar argument holds for $K'_{h}, V'_{h}$.
    As a result, the concatenated attention heads outputs
    \begin{equation*}
        \begin{pmatrix}
            Y^{(\ell, 1)} & \dots & Y^{(\ell, H)}
        \end{pmatrix} \text{.}
    \end{equation*}
    The output of the layer is then
    \begin{equation*}
        \begin{pmatrix}
            Y^{(\ell, 1)} + X^{(\ell - 1)}_{1} & \dots & Y^{(\ell, H)} + X^{(\ell - 1)}_{H}
        \end{pmatrix} 
    \end{equation*}
    so that
    \begin{equation*}
        \sum_{h} Y^{(\ell, h)} + X^{(\ell - 1)}_{h} = Y^{(\ell - 1)} + \sum_{h} Y^{(\ell, h)} = Y^{(\ell)}
    \end{equation*}
    satisfying the inductive hypothesis.
    
    At the final layer, the transformer $T'$ outputs
    \begin{equation*}
        \begin{pmatrix}
            Y^{(L, 1)} & Y^{(L, 2)} & \dots & Y^{(L, H)}
        \end{pmatrix}
    \end{equation*}
    so that $g(T'(X)) = T(X)$.
    To conclude the proof, note that $g$ is easily implemented by an \AC (and therefore \eAC) of size $O(m)$.
    Since every embedding map $Q_{h}, K_{h}, V_{h}$ takes the same input (i.e. the output of $g$), we can implement all $H$ embedding maps with an \AC of size $O(m)$.

    \paragraph{Case 2: $T$ has MLPs or GLU MLPs with $\dhid = O(m/H)$ and $\sigma$ activation.}
    If $T$ has MLPs with hidden dimension $\dhid = O(m)$ and activation function $\sigma$, we instead proceed with the inductive hypothesis 
    \begin{equation*}
        X_{1}^{(\ell - 1)} = Y^{(\ell - 1)} \text{ and } X_{h}^{(\ell - 1)} = 0_{N \times m} \text{ for $h > 1$.}
    \end{equation*}
    That is, $X_{1}^{(\ell - 1)}$ is the output of the $(\ell - 1)$-th layer of $T$ and the remaining blocks are $0$.
    The base case again holds by definition.

    Consider the $\ell$-th layer.
    Define the maps $\pi: \R^{m} \rightarrow \R^{m/H}$ as 
    \begin{equation*}
        \pi: v \mapsto v \begin{pmatrix}
            I_{m/H} \\
            0_{(m - m/H) \times (m/H)}
        \end{pmatrix} = (v_1, \dots, v_{m/H}) \text{.}
    \end{equation*}
    Furthermore, for $a \leq b$, define $i: \R^{a} \rightarrow \R^{b}$ as the map $i: v \mapsto (v, 0_{b - a})$ as the injective map with padding $0$.
    For every fixed attention head $h$, define $Q_{h}' = Q_{h} \circ \pi$ and $K_{h}', V_{h}'$ analogously.
    Here, observe that $Q_{h}', K_{h}', V_{h}'$ are embeddings that can in fact be implemented with \AC of size $O(m/H)$, so all embedding maps can be implemented with an \AC Of size $O(m)$.
    
    Now, for the MLP, if $T$ has MLPs given by $W_1: \R^{(m/H) + 1} \rightarrow \R^{\dhid}, W_2: \R^{\dhid + 1} \rightarrow \R^{m}, \sigma$, we let $T'$ have MLPs given by
    \begin{equation*}
        W_1'(X) := 
        \begin{pmatrix}
            X & 1
        \end{pmatrix} \begin{pmatrix}
            I_{m/H} & 0_{(m/H) \times 1 }\\
            \vdots & \vdots \\
            I_{m/H} & 0_{(m/H) \times 1 }\\
            0_{1 \times (m/H)} & 1
        \end{pmatrix} W_{1}, \quad W_{2}'(X) := \begin{pmatrix}
            \begin{pmatrix}
                X & 1
            \end{pmatrix} W_{2} 
            & 0_{1 \times (m - (m/H))}
        \end{pmatrix}
    \end{equation*}
    and $\sigma$ activation function.
    Note that both $W_1, W_2$ can be implemented as \eAC of size $O(m)$.

    We claim that the construction yields the desired result.
    By the inductive hypothesis, the output of the attention heads (after concatenation) in the $\ell$-th layer is
    \begin{equation*}
        Z := \begin{pmatrix}
            Y^{(\ell, 1)} & Y^{(\ell, 2)} & \dots & Y^{(\ell, H)}
        \end{pmatrix} \text{.}
    \end{equation*}
    In particular, the first linear map of the MLP $W_1'$ computes
    \begin{align*}
        W_{1}'\left( Y^{(\ell - 1)} + Z \right) &= W_{1}' \left( \begin{pmatrix}
            Y^{(\ell, 1)} + Y^{(\ell - 1)} & Y^{(\ell, 2)} & \dots & Y^{(\ell, H)}
        \end{pmatrix} \right) \\
        &= \begin{pmatrix}
            Y^{(\ell - 1)} + \sum_{h} Y^{(\ell, h)} & 1_{N \times 1}
        \end{pmatrix} W_{1} 
    \end{align*}
    is exactly the input to activation function $\sigma$ computed by $T$.
    This completes the proof of the inductive case, as the $W_{2}'$ outputs the output of the MLP into $X_{1}^{(\ell)}$ and $0$ everywhere else.
    A similar modification applies to GLU MLPs by modifying $W_{0}$ analogously to $W_{1}$.
\end{proof}

\end{document}